\documentclass[letterpaper,11pt]{article}
\usepackage[utf8]{inputenc}

\pdfoutput=1

\usepackage{booktabs} 
\usepackage[ruled]{algorithm2e} 
\usepackage{float}
\usepackage{graphicx}
\usepackage{caption}
\usepackage{subcaption}
\usepackage{amsmath, amsthm, amssymb}
\usepackage{mathtools}
\usepackage{comment} 
\usepackage[suppress]{color-edits} 
\usepackage[most]{tcolorbox}
\usepackage{xfrac}
\usepackage{hyperref}
\usepackage{multirow}
\usepackage{caption}
\usepackage{bm}
\usepackage{newfloat}
\usepackage{scalerel}
\usepackage[inline]{enumitem}
\usepackage{accents}
\usepackage{xpatch}
\usepackage{slashbox}
\usepackage{lscape}
\usepackage{etoolbox}
\usepackage{lipsum}   
\usepackage{setspace} 
\usepackage{xassoccnt}
\usepackage{refcount}
\AtBeginEnvironment{quote}{\singlespacing\small}

\SetAlFnt{\small}
\SetAlCapFnt{\small}
\SetAlCapNameFnt{\small}
\SetAlCapHSkip{0pt}
\IncMargin{-\parindent}
\def \a {{{a}}}

\def \bu {{{b}}}
\def \A {{{\sf{A}}}}
\def \B{{{\sf{B}}}}

\def\res{{{\text{\sf {R}}}}}

\usepackage{xcolor}

\renewcommand{\comment}[1]{}

\addauthor{md}{blue}
\addauthor{as}{red}

\usepackage[margin=1in]{geometry}
\allowdisplaybreaks

\definecolor{mygreen}{RGB}{20,120,60}

\hypersetup{
     colorlinks=true,
     citecolor= mygreen,
     linkcolor= mygreen
}

\title{Beating Greedy For Approximating Reserve Prices\\ in Multi-Unit VCG Auctions}
\author{Mahsa Derakhshan\thanks{Majority of the research was conducted while the author was an Intern at Microsoft Research NYC.}
 \\ University of Maryland  
\and David M. Pennock  \thanks{Majority of the research was conducted while the author was a Principal Researcher at Microsoft Research NYC.}\\ Rutgers University \and
Aleksandrs Slivkins \\ 
Microsoft Research }

\date{}
\definecolor{mygreen}{RGB}{20,140,80}
\definecolor{mylightgray}{RGB}{230,230,230}
\newcounter{myalgctr}

\newcounter{totalfootnotes}
\newcounter{totalfootnotetexts}

\DeclareAssociatedCounters{footnote}{totalfootnotes}

\usepackage{xpatch}
\usepackage{xparse}
\xpatchcmd{\footnotemark}{\stepcounter}{\refstepcounter}{}{}
\xapptocmd{\footnotemark}{\label{fnmark-\number\value{totalfootnotes}}}{}{}
\xpretocmd{\footnote}{\stepcounter{totalfootnotetexts}}{}{}

\NewDocumentCommand{\morefootnotetext}{o+m}{%
  \IfValueTF{#1}{%
    \footnotetext[#1]{#2}%
  }{%
    \stepcounter{totalfootnotetexts}%
    \footnotetext[\getrefnumber{fnmark-\number\value{totalfootnotetexts}}]{#2}
  }%
}

\newenvironment{tbox}{
\vspace{0.2cm}
\begin{tcolorbox}[width=\textwidth,
                  enhanced,
                  boxsep=2pt,
                  left=1pt,
                  right=1pt,
                  top=4pt,
                  boxrule=1pt,
                  arc=0pt,
                  colback=white,
                  colframe=black
                  ]
}{
\end{tcolorbox}
}

\newenvironment{graytbox}{
\vspace{0.2cm}
\begin{tcolorbox}[width=\textwidth,
                  enhanced,
                  frame hidden,
                  boxsep=6pt,
                  left=1pt,
                  right=1pt,
                  top=4pt,
                  boxrule=1pt,
                  arc=0pt,
                  colback=mylightgray,
                  colframe=black,
                  breakable
                  ]
}{
\end{tcolorbox}
}

\newcommand{\tboxhrule}[0]{\vspace{0.1cm} \hrule \vspace{0.2cm}}

\renewcommand{\b}[1]{\ensuremath{\bm{\mathrm{#1}}}}

\DeclareMathOperator*{\E}{\mathbb{E}}

\DeclareFloatingEnvironment[fileext=frm,name=LP]{lp}

\newcounter{eqcounter}
\addauthor{md}{blue}  

\newtheorem{theorem}{Theorem}
\newtheorem*{theorem*}{Theorem}

\newtheorem{lemma}{Lemma}[section]

\newtheorem{definition}[lemma]{Definition}
\newtheorem{claim}[lemma]{Claim}
\newtheorem{fact}[lemma]{Fact}

\newtheorem{remark}[lemma]{Remark}
\newtheorem{assumption}[lemma]{Assumption}
\newtheorem*{result*}{Main Result}

\definecolor{mygreen}{RGB}{20,100,60}

\hypersetup{
     colorlinks=true,
     citecolor= mygreen,
     linkcolor= black
}

\newcommand{\evcg}{\sf{EVCG}}

\newcommand{\opt}[0]{\ensuremath{\textsc{OPT}}}

\newcommand{\boost}[0]{\ensuremath{\beta}}

\newcommand{\rev}[2]{\ensuremath{\text{\sf{Rev}}}_{#1}(#2)}

\newcommand{\lps}[0]{\ensuremath{s^{\star}}}
\newcommand{\opts}[0]{\ensuremath{s}^{\text{o}}}

\newenvironment{titledtbox}[1]{\begin{tbox}#1
\tboxhrule
}{\end{tbox}}

\newcommand{\abs}[1]{\left\vert{#1}\right\vert}

\renewcommand{\Pr}{\operatorname{Pr}}

\newcommand{\winners}[3]{\ensuremath{W_{#1}(#2, #3)}}

\newcommand{\bid}[2]{\ensuremath{\beta_{#1, #2}}}

\begin{document}
\maketitle

\begin{abstract}
We study the problem of finding personalized reserve prices for unit-demand buyers in multi-unit \emph{eager} VCG auctions with correlated buyers. The input to this problem is a dataset of submitted bids of $n$ buyers in a set of auctions. The goal is to find a vector of reserve prices, one for each buyer, that maximizes the total revenue across all auctions.  

\smallskip
Roughgarden and Wang (2016) showed that this problem is APX-hard but admits a greedy $\sfrac{1}{2}$-approximation algorithm. 
Later, Derakhshan, Golrezai, and Paes Leme (2019) gave an LP-based algorithm achieving a $0.68$-approximation for the (important) special case of the problem with a single-item, thereby beating greedy. We show in this paper that the algorithm of Derakhshan {\em et al.} in fact {\em does not} beat greedy for the general multi-item problem. This raises the question of whether or not the general problem admits a better-than-$\sfrac{1}{2}$ approximation.

\smallskip
In this paper, we answer this question in the affirmative and provide a polynomial-time algorithm with a significantly better approximation-factor of $0.63$. Our solution is based on a novel linear programming formulation, for which we propose two different rounding schemes. We prove that the best of these two and the no-reserve case (all-zero vector) is a $0.63$-approximation.

\end{abstract}

\newpage

\section{Introduction}
The Vickrey-Clarke-Groves (VCG) mechanism, while   \emph{``lovely''} in theory~\cite{ausubel2006lovely}, is often criticized for not having a good performance guarantee when it comes to revenue maximization.  Hartline and Roughgarden~\cite{HR09} approach this shortcoming by optimizing \emph{reserve prices}, the minimum prices that the seller is willing to sell each item to each buyer. According to a number of theoretical and empirical studies \cite{HR09, ostrovsky2011reserve, beyhaghi2018improved, edelman2007internet}, {\em personalized} reserve prices (i.e., one custom reserve price for each buyer) can significantly improve revenue. In particular, personalized reserves can lead to approximately optimal revenue in quite general settings~\cite{HR09}. This has led to several studies computing an optimal vector of reserve prices \cite{beyhaghi2018improved, roughgarden2016minimizing, paes2016field, derakhshan2019lp}.
The study of VCG-like mechanisms for revenue optimization is closely aligned with a broader agenda of simple vs. optimal mechanisms ~\cite{golrezaei2017boosted, celis2014buy, paes2016field, allouah2018prior, bhalgat2012online, beyhaghi2018improved, DRY15}, and in fact has been one of the starting points for this agenda~\cite{HR09}.

\subsection{Our Scope and Main Result}

We focus on the \emph{eager} version of VCG auctions.
\footnote{An alternative, called \emph{lazy} VCG~\cite{DRY15}, first forms a set of potential winners using a VCG auction, then removes buyers whose bids don't clear their reserve. The eager version is often superior both in theory and in practice~\cite{paes2016field}.}
Consider a {\em multi-unit} auction with {\em unit-demand} buyers: that is, $k$ identical units are available, and each buyer is interested in obtaining only one unit. The auctioneer announces a reserve price $r_\bu$ for each buyer $\bu$ and then buyers place their bids. We run a VCG auction among the buyers that clear their reserve price (i.e., their bid equals or exceeds their reserve price), and the winners pay the maximum of their own reserve price and their VCG payment. That is, let $S$ be the set of buyers who clear their reserve prices. The first $k$ buyers in $S$ with the highest bids win. Each winning buyer pays the maximum of his/her reserve price and the $(k+1)$-th highest bid.

We adopt a standard data-driven model for computing reserve prices \cite{roughgarden2016minimizing, paes2016field}: given a history of buyers' bids over multiple runs of the auction, we compute a reserve price for each buyer to maximize the total revenue attained over the same dataset. 
An important property of this model is that it does not impose essentially any restrictions on the bid distributions. In particular, buyers' private values can be correlated, in contrast with other models \cite{beyhaghi2018improved, celis2014buy, DRY15, HR09} that assume independence. Moreover,
 any approximation for the data-driven model can also be used in a black-box reduction of Morgenstern and Roughgarden~\cite{morgenstern2015pseudo} to approximate the {\em Bayesian Optimization} and {\em Batch Learning} versions of the problem with (almost) the same approximation-factor.\footnote{In Bayesian Optimization, the buyers' private value distributions are independent and known to the algorithm. In Batch Learning, these distributions are unknown, and we only have access to samples drawn from them.}

In the data driven model, this problem was first studied by Roughgarden and Wang~\cite{roughgarden2016minimizing}, who showed that it is APX-hard\footnote{NP-hard to approximate better a fixed constant factor.} and gave a $1/2$ approximate greedy solution. Later, for the (important) special case of $k=1$ items, Derakhshan et~al.~\cite{derakhshan2019lp} showed that the approximation factor can be improved to $0.68$ thereby beating the greedy solution. In Section~\ref{example}, we  provide an example proving that the algorithm of Derakhshan~et~al.~\cite{derakhshan2019lp} in fact {\em does not} beat $1/2$-approximation for the general multi-item problem. This naturally raises the question:
\begin{quotation}
	{\it Is there a polynomial time algorithm beating $1/2$-approximation in the general case?}
\end{quotation}

In this work, we answer this question in the affirmative and in fact achieve a significantly better approximation ratio of  $0.63$ for any $k$ (not necessarily constant).

\newcommand{\thmmain}[0]{There exists an algorithm with running time polynomial in the input size that outputs a vector of reserve prices $\b r^o$ such that $\rev{}{\b r^o}$ is at least a 0.63 fraction of the revenue achieved from the optimal vector of reserve prices.}
\newcommand{\thmmaininformal}[0]{
Consider the data-driven model for eager VCG auctions.
There exists a polynomial-time algorithm for computing personalized reserve prices which achieves a $0.63$-approximation in expected revenue. That is, given a dataset of bids in a set of auctions $\A$, the algorithm outputs a vector of reserve prices $\b r^o$ such that the total revenue obtained from running the eager VCG mechanism on $\A$ using reserve prices $\b r^o$ is at least a $0.63$ fraction of that of any other vector of reserve prices.}
 
\begin{graytbox}
	\begin{result*}[Formalized in Theorem~\ref{theorem:mainformal}]\label{thm:main}
\thmmaininformal{}
\end{result*}
\end{graytbox}

\subsection{Our Techniques}
\label{sec:results}

 Our algorithm consists of two main parts. First, we design a polynomial-size Integer Linear Program (ILP) to describe the optimal solution. By removing the integrality constraints, we obtain a polynomial-size LP, which gives us a fractional solution. The second part is LP-rounding. Using the optimal solution of the LP, we construct two different integral solutions. We then show that the best of three vectors: these two solutions and the all-zero vector of reserve prices, is a 0.63-approximation.

To write the LP, we first need a polynomial-size representation of the solution space in which the revenue can be computed using a linear function. The natural representations (e.g., a vector of reserve prices, or a reserve price per buyer)  fail as they result in either an exponential-size solution space or a nonlinear revenue function. We come up with an alternative concise representation, based on the following  observation: to compute the revenue of an auction, we do not need to know the reserve prices of all the buyers. Rather, it suffices to only know the reserve prices of the winners and the VCG payment, which is the $(k+1)$-th highest cleared bid. The buyer who has this bid is called the \emph{supporting buyer}. We represent a solution based on its outcomes in different auctions, where the ``outcome" specifies who are the winners and who is the supporting buyer, and what are their respective reserve prices.  The revenue from each auction can be computed using a linear function based on its outcome; therefore, the overall revenue can also be written as a linear function which is sum of the revenue across all of these auctions. 

Derakhshan et al.~\cite{derakhshan2019lp} use a similar  approach for the single item case, which falls short in the general case. They capture all pertinent information about an auction in a single ``profile" which is then used to compute the revenue of the auction, and use these profiles to write an ILP. However, we need exponentially many profiles to extend this approach to the general case, essentially because we need information about all the winners in order to compute the auction's revenue. 

We proceed as follows. Instead of capturing all information about the winners in a single profile, we partition this information into several sub-profiles, each containing only a single winner and the supporting buyer. One complication is that these sub-profiles should not contradict each other (e.g., having different supporting buyers).  This issue gets even more complicated when we relax the integrality constraint of the ILP. We resolve this by introducing new variables and constraints to our LP.

Next, we use the optimal solution to the LP to construct two vectors of reserve prices which we refer to as \emph{inflated reserves} and \emph{discounted reserves}. For each buyer $\bu$, we use the LP solution to choose a threshold $t_\bu$  to determine if a reserve price is too high or too low. We construct two probability distributions: one over reserve prices above $t_\bu$, and another over reserve prices below $t_\bu$. We use these distributions to draw, resp., the inflated and discounted reserve prices. 

Let us provide some intuition for why we choose these two different vectors of reserve prices. Recall that each buyer pays the maximum of its reserve price and the VCG payment. Let us partition the winners' payments into two types: one is from winners who pay their reserve prices, and another is from those paying the VCG payment. Note that setting smaller reserve prices results in clearing more bids, and thus a larger VCG payment. 
Roughly speaking, second-type revenue from the \emph{discounted reserves}  should be larger than that of the \emph{inflated reserves},
while the opposite might hold for the first-type revenue. Intuitively, if most of the optimal revenue is from type-one payments \emph{and} high reserve prices (at least $t_\bu$ for each buyer $\bu$), then we expect the \emph{inflated reserves} to give us a high revenue. Otherwise, if the optimal revenue is a combination of type-one and type-two payments from 
small reserve prices, we expect a high revenue from the \emph{discounted reserves}. But what if the type-two payments form a substantial portion of the revenue? This is where the vector of all-zero reserve prices comes into play.  The all-zero reserve prices obtain the maximum possible revenue one can get from the type-two payments as all the bids are cleared in this case.

We analyze these three solutions simultaneously to lower bound the revenue of each solution as a function of the other two. By exploiting structural properties of our problem and the LP, we reduce the problem of finding the approximation factor to a complex non-linear optimization problem. This reduction, and solving this optimization problem, are the most technically challenging parts of the paper (see Section~\ref{section:lemma}).

To further highlight the significance of rounding our fractional solutions in two different ways, we investigate the performance of the rounding procedure by Derakhshan et al. which only outputs a single integral solution. Roughly speaking, they directly use the fractional solution of the LP as a probability distribution over the reserve prices, and for each buyer independently draw a reserve price from that. We show that for a large number of items, this approach fails to beat the greedy algorithm. More precisely, for any given constant $0<\epsilon$, we construct an example for which  the solution obtained using this rounding procedure gets at most $0.5+\epsilon$ fractional of the optimal revenue. We describe this example in Section~\ref{example}.

\subsection{Further Related Work}
Revenue maximization in multi-unit auctions is a fundamental problem in algorithmic game theory and has received  a lot of attention over the years.  In a seminal work, 
Myerson describes the game-theoretically optimal mechanism, equivalent to VCG with optimal reserve, for one item when buyers’ valuations are drawn from known independent distributions. Myerson’s mechanism is hard to directly put into practice~\cite{roughgarden2016ironing, celis2014buy, golrezaei2017boosted} as these assumptions often fail to hold in real-life applications. Further, even if they do hold, the optimal auction is still quite complex. Therefore, the need  to design simple yet efficient mechanisms has resulted in a significant body of work~\cite{golrezaei2017boosted, celis2014buy, paes2016field, allouah2018prior, bhalgat2012online, beyhaghi2018improved, derakhshan2019lp}.

 One approach taken in design of simple and efficient mechanisms is adding reserve prices to already known simple auctions such as VCG  or the second price auction (which is the special case of VCG in single-item environments.) ~\cite{paes2016field, roughgarden2016minimizing, chawla2010multi}. Moreover, reserve prices are widely used in practice and have been shown to be very important for obtaining a high revenue \cite{ostrovsky2011reserve, beyhaghi2018improved, edelman2007internet}.

The posted-price mechanism is another well-studied auction in multi-unit environments~\cite{chawla2010multi, azar2017prophet}, related to a line of work on prophet inequalities~\cite{correa2019prophet, correa2017back}. For the case of independent buyer valuations, this mechanism has a particularly strong performance guarantee for large numbers of items~\cite{yan2011mechanism, beyhaghi2018improved}.
To the best of our knowledge, the multi-unit posted-price mechanism has not been studied for correlated buyers.

\section{Preliminaries and Problem Statement} \label{sec:pre}
There are $k$ identical items and $n$ unit-demand buyers participating in a set of eager VCG auctions. Let $\A$ and $\B$ respectively be the set of auctions and buyers. We are given a dataset of bids $\b \beta$ where for any auction $\a\in \A$ and buyer $\bu\in \B$,  $\beta_{\a, \bu}$ represents bid of buyer $\bu$ in auction $\a$. 
Let $r_{\bu}$ be the personalized reserve price of buyer $\bu$. Then, given
the bid vector $\b \beta_{\a}$ for auction $\a$ and  reserve
price vector $\b r$, the eager VCG auction ($\evcg$) 
works as follows. 
\begin{enumerate}
\item Any buyer $\bu$ with $\beta_{\a, \bu} < r_{\bu}$ is
  eliminated. Let $S_{\a}=\{\bu: \beta_{\a, \bu}\ge r_{\bu}\}$ be the set of buyers who clear their reserve prices in auction $\a$.
\item An item is allocated to a buyer $\bu$ if there are at most $k-1$ buyers in $S_{\a}$ whose bid is greater than $\beta_{\a, \bu}$.
 \item Pick the \emph{supporting buyer} $\bu_s$ to be a buyer in set $S_{\a}$ such that there are exactly $k$ buyers in $S_{\a}$ whose bid in auction $\a$ is greater than $\beta_{a, \bu_s}$.\footnote{As usual, we assume that no two buyers have the same bids, or we can break ties based on their IDs.} ($\beta_{\a, \bu_s}$ is the VCG payment of any winner.) 
\item Any buyer $\bu$ who receives an item is charged $\max{(r_{\bu}, \beta_{\a, \bu_s})}$, otherwise they are not charged.

\end{enumerate}

Given the dataset of bids $\b \beta$, our goal is to find a vector of personalized reserve prices that maximize revenue of the auctioneer. Note that the reserve prices are the same across all the auctions $\a \in \A$. However, each buyer $\bu$ is assigned a personalized reserve price $r_{\bu}$. 
We  assume, w.l.o.g. that the optimal reserve price for any buyer is equal to one of their submitted bids. Let $\res = \{\beta_{\a, \bu}: \a\in \A, \bu\in \B\}$. 
 Formally, we would like to solve the following optimization problem:
 \begin{align}\nonumber
   &\evcg^\star~=~\max_{\b r \in \res^n }~ \rev{}{\b r},  
   \quad\text{where}\quad
    \rev{}{\b r} := 
    \textstyle \sum_{\a \in\A}\;\rev{\a}{\b r}
     \nonumber
 \end{align}
 and $\rev{\a}{\b r}$ is the total payment in action $\a$ given the vector of reserve prices $\b{r}$. Note that to solve this problem we face a search space of  size {$\abs{\res}^{n}$} which is exponential in the input size.

\section{The Algorithm}
In this section we provide an LP-based algorithm for finding a vector of personalized reserve prices given a dataset of bids. We observe that to be able to describe the optimal solution of the problem using polynomially many linear constraints, we need a concise representation of the solution space. In Section~\ref{section:solution-space}, we explain this representation. In Section~\ref{section:lp}, we use this representation to design an LP and prove that its  objective function in its optimal solution is an upper bound for the revenue of the optimal solution of the problem. Finally, in Section~\ref{section:rounding}, we provide our rounding procedure that uses the optimal solution of the LP and outputs a vector of reserve prices. 
\subsection{An Alternative Solution Space}
\label{section:solution-space}
In this section, our goal is to give a concise representation of the solution space which will help us to write our linear program. We need this representation to have a polynomial size and it should be possible to compute its revenue using a linear function. As mentioned before, we base our design on the observation that to compute the revenue of an auction we do not need to have the reserve price of all the buyers. Rather, it only suffices to know 
the reserve prices of the winners and bid of the supporting buyer. The main idea here is to have variables that capture the outcome of auctions (i.e., who the winners and the supporting buyer are and their reserve prices.) instead of just having variables for reserve price of buyers. We will define \emph{valid profiles} and \emph{valid sub-profiles} of an auction to capture its outcome. Roughly speaking, the revenue obtained from each auction can be computed in a linear way based on its outcome; thus the overall revenue can also be written as a linear function which is sum of the revenue across all these auction. 

\begin{definition}[Valid Profiles]\label{def:profile}
  {We define the set of valid profiles of an auction $\a$ as the set
  $\mathcal P_\a$ consisting of all tuples $(\bu_1, \dots, \bu_{k+1}, r_1, \dots, r_{k+1}) \in \B^{k+1} \times \res^{k+1}$}
   that satisfy the following conditions:
\begin{enumerate}
\item For any $i, j\in [k+1]$ where $i < j$, bid of buyer $\bu_i$ is greater than or equal to that of buyer $\bu_j$ in auction $\a$;
  that is, $\beta_{\a, \bu_i} \ge \beta_{\a,\bu_j}$. 
\item For any $i\in [k+1]$ buyer $\bu_i$ clears his reserve, $r_i$, in auction $\a$; that is, $\beta_{\a,\bu_i} \geq r_i$.
\end{enumerate}
\end{definition}

Valid profiles are defined to capture the outcome of an auction given a set of reserve prices. However, note that each profile consists of at least $k+1$ buyers; thus, to be able to capture all the possible scenarios (for example when fewer than $k+1$ buyers clear their reserves in an auction), we add $k+1$ auxiliary buyers $\hat{\bu}_1 \dots \hat{\bu}_{k+1}$ to $\B$ who bid zero in all the auctions. Also, w.l.o.g., we  assume that their reserves are always set to zero as well.

As mentioned previously, Derakhshan et al.~\cite{derakhshan2019lp}, use a concept similar to valid profiles to write their linear program for $k=1$. In their LP,  there is a variable for any pair of auction and valid profile. However, it does not work for us since it results in having exponentially many variables. To overcome this, we define \emph{valid sub-profiles} of an auction that only contains information about a single winner and the supporting buyer as defined below.

\begin{definition}[Sub-profiles]\label{def:profile}
  {We define the set of valid Sub-profiles of an auction $\a$ as the set
  $\mathcal S_\a$ consisting of all tuples $(\bu_1, \bu_2, r_1, r_2) \in \B^2 \times \res^2$}
   that satisfy the following conditions:
\begin{enumerate}
\item Bid of buyer $\bu_1$ is greater than or equal to that of buyer $\bu_2$ in auction $a$;
  that is, $\beta_{\a, \bu_1} \ge \beta_{\a,\bu_2}$. 

\item Buyers $\bu_1$ and $\bu_2$ clear their reserves in auction $\a$ if reserve prices $r_1$ and $r_2$ are set for them respectively; that is, $\beta_{\a,\bu_1} \geq r_1$ and $\beta_{\a,\bu_2} \geq r_2$.
\end{enumerate}
For any given $p=(\bu_1, \bu_2, r_1, r_2) \in \mathcal S_\a$, we have  $\rev{\a}{p}
  := \max(\beta_{\a, \bu_2}, r_1)$.
\end{definition}
Given a vector of reserve prices $\b r'$, we say a sub-profile $(\bu_1, \bu_2, r_1, r_2) \in \mathcal S_\a$ happens in auction $\a$ after applying $\b{r}'$, iff $r'_{\bu_1} = r_1$, $r'_{\bu_2} = r_2$, buyer $\bu_1$ is a winner in auction $\a$ and buyer $\bu_2$ is the supporting buyer in this auction. Moreover, we say two sub-profiles $(\bu_1, \bu_2, r_1, r_2)$ and $(\bu'_1, \bu'_2, r'_1, r'_2)$ are compatible iff $\bu'_2 = \bu_2$ and $r'_2 = r_2$ which means that they have the same information about the supporting buyer. Moreover, we say a set $P$ of valid sub-profiles are compatible iff they are pairwise compatible and $|P|=k$.
Since we have added $k+1$ auxiliary buyers whose bid is always cleared in all the auction, we can assume that we always have exactly $k$ winners and a supporting buyer. 

To explain how a solution is represented using these sub-profiles, we consider a vector of reserve prices $\b r$ and construct its representation in this new solution space. For any auction $\a$ and any sub-profile $p\in \mathcal S_\a$, we have a variable $s_{\a, p}$ which is equal to one iff sub-profile $p$ happens in auction $\a$ after applying vector of reserve prices $\b r$. Otherwise we have $s_{\a, p}=0$. We say vector $\b s$ constructed in this way is the representation of $\b r$ in the profile space. As mentioned above,  this representation allows us to compute the revenue of each auction using a linear function.  Recall that for any sub-profile $p= (\bu_1, \bu_2, r_1, r_2)$ we have $\rev{\a}{p}
  := \max(\beta_{\a, \bu_2}, r_1)$, thus we can write
$$\rev{\a}{\b r} = 
\textstyle \sum_{p\in \mathcal{S}_{\a}}\; s_{\a,p} \cdot \rev{\a}{p}.$$
This function is linear since we have polynomially many valid sub-profiles; thus, for any valid sup-profile $p$ we can simply compute $\rev{\a}{p}$ in advance  and treat it as a constant in the LP.

\subsection{The Linear Program}
\label{section:lp}

In this section we first design an integer linear program (ILP) then remove its integrality constraints to get an LP. We start by introducing the variables of our ILP. We have four vectors of random variables $\b s$, $\b x$, $\b y$ and $\b y'$ as defined below. 
\begin{enumerate}
\item 
For any auction $\a\in \A$ and any sub-profile $p\in \mathcal S_{\a}$, we have a variable $s_{\a, p} \in \{0, 1\}$ which is equal  to one iff sub-profile $p$ happens in auction $\a$. This set of variables should satisfy constraint $\sum_{p \in \mathcal S_{\a}} s_{\a,
    p} \leq k$ as at most $k$ sub-profiles can happen in an auction. 
\item  For any buyer $\bu\in \B$ and any reserve price $r\in \res$ we have a variable $x_{\bu, r} \in \{0, 1\}$. Reserve price $r$ is assigned to buyer $\bu$ iff $x_{\bu, r}=1$. For this type of variables we enforce the necessary constraint  $\sum_{r \in \res} x_{\bu, r} = 1$ in our LP since each buyer has exactly one reserve price. 
\item  For any buyer $\bu\in \B$, any auction $\a$ and any reserve price $r\in \res$, we have a  variable
  $y_{\bu, r, a} \in \{0, 1\}$ that is
     equal to one iff buyer $\bu$ is assigned a
    reserve price of $r$ as a winner in auction $\a$.    
 \item  For any buyer $\bu\in \B$, any auction $\a$ and any reserve price $r\in \res$, we have a variable
  $y'_{\bu, r, a} \in \{0, 1\}$ that is equal to one iff buyer $\bu$ is assigned a reserve price of $r$ as the supporting buyer in auction $\a$.    
\end{enumerate}

Roughly speaking, the vector of variables $\b x$ is used in the LP to ensure that in different auctions, solution $\b s$ does not assign different reserve prices to the same buyer. Moreover variables  $\b y$ and $\b y'$ are to ensure that the sub-profiles in an auction are compatible with each other.

To be able to write the constraints of our LP, we first need the following definitions. Let $\mathcal{Q}_{\bu, \a} := \{(\bu, \bu_2, r_1, r_2)\in \mathcal{S}_\a\ | \bu_2 \in \B, r_1\in \res , r_2\in \res \}$ denote the set of valid sup-profiles of auction $\a$ in which buyer $\bu$ is the winner and $\mathcal{Q}'_{\bu, \a}= \{(\bu_1, \bu, r_1, r_2)\in \mathcal{S}_\a\ | \bu_1 \in \B, r_1\in \res , r_2\in \res \}$ is the set of valid sub-profiles of auction $\a$ in which buyer $\bu$ is the supporting buyer. 
Moreover, let us define $\mathcal{Q}_{\bu,r,\a} := \{ (\bu, \bu_2, r, r_2) \in \mathcal{Q}_{\bu, \a} | \bu_2 \in \B, r_2\in \res \}$ and $\mathcal{Q'}_{\bu,r,\a} := \{ (\bu_1, \bu, r_1, r) \in\mathcal{Q}'_{\bu, \a} | \bu_1 \in \B, r_1\in \res\}.$ We are now ready to write our LP (ILP without the integrality constraints) which we present in Figure~\ref{fig:LP}.

\begin{figure}[hbt]
    \begin{align} \label{fig:LP}
		&\max_{\b x, \b s} &&\ \sum_{\a\in \A}\sum_{p\in \mathcal{S}_{\a}} s_{\a,p} \cdot \rev{\a}{p}\nonumber	 &&\\\tag{1}
		&\,\text{s.t.} \qquad 
		& & y_{\bu, r, \a} = \sum_{\mathclap{p\in \mathcal{Q}_{\bu ,r{,\a}}}} s_{\a,p} && \forall{\a, \bu, r}: \bu \in \B, \a \in \A,  r \in \res \\\tag{2} 
		& \qquad 
		& & y'_{\bu, r, \a} = \sum_{\mathclap{p\in \mathcal{Q'}_{\bu ,r,\a}}} \frac{s_{\a,p}}{k} && \forall{\a, \bu, r}: \bu \in \B, \a \in \A,  r \in \res \\\tag{3}
				& \qquad 
				& &  y_{\bu, r, \a} + y'_{\bu, r, \a} \leq x_{\bu, r} && \forall{\a, \bu, r}: \bu \in \B, \a\in \A, r \in \res ,  \\\tag{4}  
		    & &&  \sum_{\mathclap{\substack{p\in \mathcal{Q}_{\bu_2, {\a}} \\ \cap \mathcal{Q'}_{\bu_1, {\a}} }}} s_{\a,p} \leq \sum_{\mathclap{r\in \res}} y'_{\bu_1,r, \a} \qquad && \forall{\a, \bu_1, \bu_2}: \bu, \bu_2 \in \B, \a\in \A\\ \tag{5}
		& \qquad 				
		& & \sum_{p \in \mathcal{S}_{\a}} s_{\a,p} \leq k && \forall{\a}: \a \in \A\\\tag{6}
		& && \sum_{r\in \res} x_{\bu,r} = 1 && \forall{\bu}: \bu \in \B	\\
		&&& s_{\a,p} \geq 0  &&  \forall{\a, p} :  \a \in \A, \ p \in \mathcal{S}_{\a} \tag{7} 
  \end{align}
  
  \caption{The Linear Program}
\end{figure}

In the rest of the paper we use $\b \lps$ to refer to $\b s$ from an optimal solution of the LP. To be able to use the optimal solution of the LP as a benchmark in analyzing  the approximation-factor of our algorithm, we need to show that it is indeed an upper bound for the optimal integral solution.
\newcommand{\lemmaLPOpt}{The optimal revenue is upper bounded by $\sum_{\a\in \A}\sum_{p\in \mathcal{S}_{\a}} \lps_{\a,p} \cdot \rev{\a}{p}$.
}

\begin{lemma}\label{lemma:LP-upper}
\lemmaLPOpt
\end{lemma}
We have Section~\ref{proof:lemma-lp-upper} designated to the formal proof of this lemma, and also give an informal overview of that here. Roughly speaking, to prove this lemma, it suffices to show that the constraints of the LP are all necessary for the consistency of the variables that we have defined. Below we give some intuition about  each constraint and why it is necessary.

 Constraint (1) is due to the fact that for any auction $\a$, buyer $\bu$, and reserve price $r$, variable $y_{\bu, r, \a}$ indicates whether or not buyer $\bu$ has a reserve $r$ and is a winner in auction $\a$. This constraint ensures that  value of $y_{\bu, r, \a}$ is consistent with whether or not there is a profile happening in auction $\a$ in which buyer $\bu$ is a winner as is assigned a reserve price of $r$. Constraint (2)  is similar to the previous one but for the supporting buyers. Constraint (3) is because the reserve prices assigned to a buyer in different auctions should be consistent.  Moreover, imposing constraint (4), on the sub-profiles is to make sure that the sub-profiles that happen in an auction can form valid profiles. 
  Consider an auction $\a$ and buyers $\bu_1$ and $\bu_2$. The right-hand-side of this constraint is the probability with which buyer $\bu_1$ is the supporting buyer in auction $\a$, and the left-hand-side is the probability with which buyer $\bu_2$ is a winner while buyer $\bu_1$ is the supporting buyer which should obviously be smaller than the probability that buyer $\bu_2$ is a supporting buyer. Finally, constraints (5), (6), and (7) are by definition of variables. 
This LP upper bounds the optimal solution (the proof is deferred to the appendix).

\subsection{The LP-Rounding Procedure}
\label{section:rounding}
In this section, given an optimal solution of the LP we generate an integral solution for the problem. Input of the algorithm is  the vector $\b x$ from an optimal solution  of the LP and a parameter $ \boost{} \in [0, 1]$, which we fix later.

\begin{enumerate}[leftmargin=*]
\item For any buyer $\bu$ let $t_\bu$ be the maximum number in $\res$ that satisfies $\sum_{r < t_{\bu}} x_{b, r} \leq \boost$. 
\item Define vectors $\b f$ and $\b f'$ as follows: For any $r\in \res$ where $r < t_{\bu}$ set $f_{b, r} := x_{b, r}/\boost$ and $f'_{b, r} := 0$. For any $r > t_{\bu}$, set $f_{b, r} := 0$ and $f'_{b, r} := x_{b, r}/(1-\boost)$. Finally for $r=t_b$, set $f_{b, r} := 1 - \sum_{r' < t_b} f_{b, r'}$ and $f'_{b, r} := 1 - \sum_{r' > t_b} f'_{b, r'}$.

\item Construct $\b r$ the vector of \emph{discounted} reserve prices as follows: For any buyer $\bu$ independently choose a random reserve price $r_\bu\in \res$ such that for any $\rho\in \res$, we have $\Pr[r_\bu= \rho] =f_{\bu, \rho}$.
\item Construct $\b r'$ the vector of \emph{inflated} reserve prices as follows: For any buyer $\bu$ independently choose a random reserve price $r'_\bu\in \res$ such that for any $\rho\in \res$ we have $\Pr[r'_\bu= \rho]=f'_{\bu, \rho}$.
\item Let $\b z$ be the vector of all zero reserve prices.
\item Between $\b z$, $\b r$ and $\b r'$ return the one with higher revenue which is as follows: $$\arg \max_{\mathclap{\b \nu \in \{\b z, \b r, \b r'\}}} \; \rev{}{\nu}. $$
\end{enumerate}
In the rest of the paper we use $\b r$ and $\b r'$ to respectively refer to the vector of discounted and inflated reserve prices constructed in this algorithm. 
\begin{remark}
For the sake of simplicity in the analysis, we assume w.l.o.g., that for any auction $\a$, $t_\bu$ satisfies $\sum_{r < t_{\bu}} x_{b, r} = \boost$. This implies that for any $r\in \res$, if $r<t_b$ we have $\Pr[r_b = r] = x_{\bu, r}/\boost$ and $\Pr[r'_b = r] = 0$. Otherwise, if $r \geq t_b$, we have $\Pr[r_b = r] = 0$ and $\Pr[r'_b = r] = x_{\bu, r}/(1-\boost)$.  
\end{remark}
 
In the next section we show how we can use specific features of the three solutions $\b z$, $\b r$ and $\b r'$ to get our desired approximation-factor.

\section{Approximation Factor}
 
 In this section, we prove our main theorem by giving a lower bound for the revenue obtained from the vector of reserve prices outputted by the rounding algorithm. 
 
 Let us start by giving some definitions that will be used throughout this section. Define $\beta^{(k+1)}_{\a}$ to be the $(k+1)$-th highest bid in any auction $\a$. Note that this is different from the bid of the supporting buyer in auction $\a$ as the supporting buyer has the $(k+1)$-th highest bid after removing the buyers whose bid is not cleared. Moreover, given a threshold $\tau$, let us denote by $\winners{\a}{\b r}{\tau} $ the number of winners in auction $\a$ whose payment is greater than or equal to $\tau$ using the vector of reserve prices $\b r$. Similarly, $\winners{\a}{\b r'}{\tau}$ is the number of winners who pay at least $\tau$ using vector of reserve prices $\b r'$ in auction $\a$. Note that $\winners{\a}{\b r}{\tau}$ and $\winners{\a}{\b r'}{\tau}$ are both random variables.

The following lemma establishes sufficient conditions for the algorithm to output an approximate solution. For any given auction $\a\in \A$ and a real number $\tau \geq 0$, we define
\begin{align*} 
\Phi_a(\tau) :=   
\;\;\sum_{\mathclap{\substack{
p : p\in \mathcal{S}_{\a}, \\ \;\; \rev{\a}{p} \geq \tau}}} \;
\lps_{\a, p} - (1-\beta)\E[\winners{a}{\b r'}{\tau}]   - \beta \E[\winners{\a}{\b r}{\tau}]. 
  \end{align*}
  
  \begin{lemma}\label{lem:twocond} 
Suppose there exist absolute constants $\boost\in (0, 1)$ and $c\in (0, 1)$ such that, 
\begin{align} \label{eq:lemma1} 
\Phi_a(\tau)\leq 0
    &\qquad\text{if}\quad 
    \tau > \beta^{(k+1)}_{\a} \text{, and }\\
    \label{eq:lemma2}
\Phi_a(\tau)\leq kc
    &\qquad\text{if}\quad 
    \tau \le  \beta^{(k+1)}_{\a}
\end{align}
for any auction $\a\in \A$.
Then, the algorithm with parameter $\boost$ outputs a $\frac{1}{1+c}$-approximate solution.
\end{lemma}
\begin{proof}
   Consider an arbitrary auction $\a$. By integrating over $\tau$ in $\Phi_a(\tau)$, we obtain:
     \begin{align} 
    \nonumber &\int
    \limits_{(\beta^{(k+1)}_{\a}, \infty)}
    \Phi_a(\tau)\;
    d\tau \; +\; 
     \int\limits_{(0, \beta^{(k+1)}_{\a}]}
     \Phi_a(\tau)\;
     d\tau \leq k c \beta^{(k+1)}_{\a}.
      \end{align}
By simplifying this we get 
\begin{equation}\label{eq:revenuelem} \sum_{p\in \mathcal{S}_{\a}} s_{\a,p} \cdot \rev{\a}{p} - (1-\beta) \E[\rev{\a}{\b r'}] - \beta \E[\rev{\a}{\b r}] \leq k c \beta^{(k+1)}_{\a}.\end{equation}
Recall that the output of our algorithm is the best of $\b r$, $\b r'$, and $\b z$ where $\b z$ is the vector of all-zero reserve prices and its revenue is $k\beta^{(k+1)}_{\a}$ since by applying that, the players in the $k$ first positions win the $k$ items and pay bid of the buyer in the $(k+1)$-th position.  Therefore, the expected revenue achieved from the output of our algorithm is at least $$\mu := \max\left(\E[\rev{}{\b r}], \E[\rev{}{\b r'}], k\cdot \sum_{\a\in \A}\beta^{(k+1)}_{\a}\right).$$  
 Based on Equation~\ref{eq:revenuelem} we have 
 $$\rev{}{\b \lps} - (1-\beta) \mu - \beta \mu - c \mu \leq 0.$$
 where $\rev{}{\b \lps} := \sum_{\a\in \A}\sum_{p\in \mathcal{S}_{\a}} \lps_{\a,p} \cdot \rev{\a}{p}.$ 
 This implies  $\rev{}{\b \lps} - (1+c)\mu \leq 0$, and as a result $$\frac{\rev{}{\lps}}{1+c} \leq  \max\left(\E[\rev{}{\b r}], \E[\rev{}{\b r'}], k\beta^{(k+1)}_{\a}\right).$$ Further by Lemma~\ref{lemma:LP-upper}, we know that $\rev{}{\lps}$ is an upper bound for the revenue of the optimal solution; thus by setting $\boost = \boost'$ in the rounding algorithm its output is at least a $\frac{1}{1+c}$-approximate solution.
 \end{proof}
 
  Having, Lemma~\ref{lem:twocond}, it suffices to  prove that Equation~\ref{eq:lemma1} always holds and find values for parameters $\beta$ and $c$ that satisfy Equation~\ref{eq:lemma2}. We address the former in the following lemma, and the latter in Lemma~\ref{lem:second-type}.
  
 \begin{lemma}\label{lemma:first-type}
For auction $\a$ and any $\tau > \beta_{\a}^{(k+1)} $ we have 

$$\sum_{\mathclap{\substack{p : p\in \mathcal{S}_{\a}\\ \rev{\a}{p} \geq \tau}}} \lps_{\a, p} - (1-\beta)\E[\winners{a}{\b r'}{\tau}]   - \beta \E[\winners{\a}{\b r}{\tau}] \leq 0. $$ 
\end{lemma}

\begin{proof}
By definition of $\winners{\a}{\b r'}{\tau}$ and $\winners{\a}{\b r}{\tau}$ for any $\tau > \beta_{\a}^{(k+1)} $, we have $$(1-\beta)\E\left[\winners{a}{\b r'}{\tau}\right]   +\beta \E\left[\winners{\a}{\b r}{\tau}\right] =  \left(1-\beta\right) \E\left[ \min\left(\sum_{\bu \in \B} 1_{r'_b > \tau}, k \right)\right] + \beta \E\left[ \min\left(\sum_{\bu \in \B} 1_{r_b > \tau}, k \right)\right].$$
Recall that $\beta^{(k+1)}_{\a}$ is defined in a way that there are exactly $k$ buyers with bids greater than $\beta^{(k+1)}_{\a}$ in auction $\a$. As a result, there are at most $k$ buyers for whom $\Pr[r'_b > \tau]$ or $\Pr[r_b > \tau]$ is nonzero. This means that we can rewrite the inequality as $$(1-\beta)\E\left[\winners{a}{\b r'}{\tau}\right]   +\beta \E\left[\winners{\a}{\b r}{\tau}\right] =  \left(1-\beta\right) \E\left[ \sum_{\bu \in \B} 1_{r'_b > \tau}\right] + \beta \E\left[\sum_{\bu \in \B} 1_{r_b > \tau}\right] = \E\left[\sum_{\bu \in \B} 1_{r_b > \tau}\right].$$
Observe that, by construction of \b{r} and \b{r'} we have \begin{equation} \label{eq:hahah}\E[1_{r_b>\tau}]=\Pr[r_\bu > \tau ] = \frac{1}{\beta}\sum_{r\in (\tau, t_\bu)} x_{\bu, r} \;\; \text{ and } \;\; \E[1_{r'_b>\tau}]=\Pr[r'_\bu > \tau ] = \frac{1}{1-\beta}\sum_{r: r \geq  t_\bu, r>\tau} x_{\bu, r}.\end{equation}
Moreover, by putting the first and  third constraints of the LP together, for any buyer $\bu$ and any reserve price $r > \tau$, we get $\sum_{p\in Q_{\bu, r, \a}} \lps_{\a, p} \leq x_{\bu, r}.$  Note that for any $p=(\bu_1, \bu_2, r_1, r_2) \in \mathcal{S}_\a$, by definition, we have  $ \beta_{\a}^{(k+1)} \geq \bid{\a}{\bu_2}$, which means that if we have $\rev{\a}{p} >\tau$,  then $\rev{\a}{p} = r_1$ and $p \in Q_{\bu_1, r_1, \a}$. This results in the following equations.  $$\sum_{\bu\in \B}\;\;\; \sum_{\mathclap{r\in (\tau, t_\bu)}} x_{\bu, r} \geq \sum_{\mathclap{\substack{p: p\in \mathcal{S}_{\a},\\ \rev{\a}{p} <  t_\bu, \\ \rev{\a}{p} > \tau}}} \lps_{\a, p},\;\; \text{ and } \;\;\; \sum_{\bu\in \B}\,\;\; \sum_{\mathclap{\substack{r: r\geq t_\bu, \\ r>\tau}}} x_{\bu, r} \geq \sum_{\mathclap{\substack{p: p\in \mathcal{S}_{\a},\\ \rev{\a}{p} \geq t_\bu,\\ \rev{\a}{p} > \tau}}} \lps_{\a, p}.$$
Combining these with the equations in (\ref{eq:hahah}) gives us the following equation and concludes the proof.
  $$(1-\beta)\E\left[\winners{a}{\b r'}{\tau}\right]   +\beta \E\left[\winners{\a}{\b r}{\tau}\right] \geq \sum_{\substack{p: p\in \mathcal{S}_\a\\ \rev{\a}{p}>\tau}} \lps_{\a, p}.$$ 
\end{proof}
The following lemma is the most technically challenging part of the paper. Therefore, we have Section~\ref{sec:proof} (almost the rest of the paper) assigned to its proof. 
 \begin{lemma}\label{lem:second-type}
  
Setting $\beta = 0.55$ and $c=0.58$, the following inequality holds for any auction $\a \in \A$ and any  $0<\tau \le  \beta^{(k+1)}_{\a}$.
  \begin{align}   \sum_{\mathclap{\substack{p : p\in \mathcal{S}_{\a},\\ \rev{\a}{p} \geq \tau}}} \lps_{\a, p}  - (1-\beta)\E[\winners{\a}{\b r'}{\tau}]   - \beta \E[\winners{\a}{\b r}{\tau}] ~ &\leq~
  kc.\label{eq:condition2}\end{align}  \end{lemma}
 
We are now ready to prove the main result of the paper. Below we restate the main theorem and prove it using the lemmas in this section.
   
\begin{theorem}\label{theorem:mainformal}
\thmmain{}	
\end{theorem}

\begin{proof}
First, note that the LP designed in Section~\ref{section:lp} has polynomially many variables and constraints. To design our algorithm, we first solve the LP, then given an optimal solution of that use the LP-rounding procedure to output a vector of reserve prices. Since the LP rounding procedure has a polynomial running time, the total running time of the algorithm is polynomial as well. 

To analyze the approximation factor of the algorithm we use Lemma~\ref{lem:twocond}, Lemma~\ref{lem:second-type} and Lemma~\ref{lemma:first-type}. The first lemma states that if there exists a constant $c\geq 0$ and a valuation for parameter $\beta$ that satisfies Equation~\ref{eq:lemma1} for any $\tau > \beta^{(k+1)}_{\a}$, and satisfies Equation~\ref{eq:lemma2} for any $\tau \leq \beta^{(k+1)}_{\a}$, then our rounding algorithm is a $\frac{1}{1+c}$-approximation. In Lemma~\ref{lem:second-type}, we prove that Equation~\ref{eq:lemma1} holds for any $\beta \in (0, 1)$ and any $\tau > \beta^{(k+1)}_{\a}$. Moreover,   based on Lemma~\ref{lemma:first-type} we  have that by setting $\boost = 0.55$,  Equation~\ref{eq:lemma2} holds for any $\tau \leq \beta^{(k+1)}_{\a}$ and $c=0.58$. This implies that by setting $\boost = 0.55$ in the rounding algorithm, its output is a  0.63-approximation of the optimal solution since $\frac{1}{1+0.58} > 0.63$. 
\end{proof}

\section{Proof of Lemma~\ref{lem:second-type}}
\label{section:lemma}
\label{sec:proof} 
In this section, we will consider an arbitrary auction $\a \in \A$ and any constant $0< \tau \leq \beta^{(k+1)}_\a$,  and focus on finding a constant $c$ and a valuation for $\boost$ that satisfy \begin{align}\nonumber   \sum_{\mathclap{\substack{p : p\in \mathcal{S}_{\a},\\ \rev{\a}{p} \geq \tau}}} \lps_{\a, p}  - (1-\beta)\E[\winners{\a}{\b r'}{\tau}]   -  kc\leq \beta \E[\winners{\a}{\b r}{\tau}].\end{align}
  Denote by $\B_{\a, \tau} $ the set of buyers whose bid in auction $\a$ is greater than or equal to $\tau$. Formally, we have $\B_{\a, \tau} := \{\bu\in \B: \bid{\a}{\bu} \geq \tau\}.$
 Throughout this section, since we assume that $\a$ can be any arbitrary auction from $\A$, we will abbreviate all notations by dropping $\a$ for simplicity when clear from the context. 
Let us define function $F(\b \lps, \tau)$ as follows. 
\begin{align}\label{eq:defn-F}
F(\b \lps, \tau) = \sum_{\mathclap{\substack{p : p\in \mathcal{S}_{\a},\\ \rev{\a}{p} \geq \tau}}} \lps_{\a, p}  - (1-\beta)\E[\winners{\a}{\b r'}{\tau}]
\end{align}
  We will consider different values of $F(\b \lps, \tau)$ as a function of $\beta$ and  give a lower bound for $\E[\winners{\a}{\b r}{\tau}]$ based on that. 
Consider a buyer $\bu\in \B_{\a, \tau}$. Let us define Bernoulli random variables $p_{\bu, \tau}$ and $q_{\bu}$ to be respectively equal to one iff $r_\bu\in [\tau, \bid{\a}{\bu}]$ and equal to one iff $r_{\bu}\in [0, \bid{\a}{\bu}]$. Moreover, let $P_{\tau}=\sum_{\bu\in \B_{\tau}} p_{b,\tau}$ and $Q_{\tau}=\sum_{\bu \in \B_{\tau}} q_{\bu}$. 
\begin{claim}\label{claim:exprevenue-r}
The expected revenue obtained from the vector of reserve prices $\b r$ is as follows.
 \begin{align} \E[\winners{\a}{\b r}{\tau}] \geq \max\big(k.\Pr[Q_\tau>k], \E[\min(P_\tau, k)] \big)	 \nonumber \end{align}
\end{claim}

\begin{proof}
	Note that $Q_\tau$ is a random variable representing the number of buyers whose bid is cleared and is greater than or equal to $\tau$; therefore, $Q_\tau>k$ is the event in which at least $k+1$ buyers have cleared their bid of at least $\tau$ which results in $k$ items being sold with a price of at least $\tau$.  Moreover, $P_\tau$ denotes the number of buyers whose bid is cleared with a reserve of at least $\tau$; thus, we sell at least $\min(P_\tau, k)$ of our items with a price of at least $\tau$.  This means that the expected number of items that are sold with a price of at least $\tau$ is lower bounded by  $\max(k.\Pr[Q_\tau>k], \E[\min(P_\tau, k)]).$
	\end{proof}

Let $\mathcal{T}_{\a, \tau}$ be the set of sub-profiles in $\mathcal{S}_{\a}$ whose revenue is at least  $\tau$. In the other words, 
$$ \mathcal{T}_{\a, \tau} := \{p\in \mathcal{S}_{\a} |\rev{\a}{p} \geq \tau \}.$$
We partition $\mathcal{T}_{\a, \tau}$  to three disjoint subsets  denoted by $\mathcal{J}^+_{\tau}, \mathcal{J}^{-}_{ \tau}$, and $ \mathcal{L}_\tau$ as follows.
Set $\mathcal{J}^{+}_{\tau}$ is the set of sub-profiles in $\mathcal{T}_{\tau}$ that capture the scenarios in which the supporting buyer has a bid smaller than $\tau$ and the winner's reserve price is greater than or equal to its threshold $t_{\bu}$ (defined in the algorithm).
$$ \mathcal{J}^{+}_{\tau} := \{p=(\bu', \bu, r', r)\in \mathcal{T}_{\tau} | \, \bu \notin \B_\tau \text{ and } r' \geq t_{b'}\}.$$
We similarly define $\mathcal{J}^{-}_{\bu, \tau}$ to be the set of sub-profiles in $\mathcal{T}_{\tau}$ in which the supporting buyer has a bid smaller than $\tau$ and the reserve price of the winner is below its threshold $t_\bu$. 
$$ \mathcal{J}^{-}_{\tau} := \{p=(\bu', \bu, r', r)\in \mathcal{T}_{\tau} | \, \bu \notin \B_\tau  \text{ and } r' < t_{b'}\}. $$
Moreover, $\mathcal{L}_{\tau}$ defined below denotes the set of sub-profiles in $\mathcal{S}_{\a}$ that capture the scenarios in which the supporting buyer has a bid greater than or equal to $\tau$.
$$\mathcal{L}_{ \tau} := \{p=(\bu', \bu, r', r)\in \mathcal{T}_{\tau} | \,\bu \in \B_\tau \}.$$
Further, based on this set, we define 
\begin{align}\label{eq:defn-delta-tau}
\delta_\tau := \sum_{\bu\in \B_\tau} \sum_{p \in \mathcal{L}_{\tau} \cap \mathcal{Q'}_{\bu}} \lps_{\a,p}/k.
\end{align}
Observe that the defined subsets of $\mathcal{T}_\tau$ satisfy $ \mathcal{T}_\tau= \mathcal{J}^{+}_{\tau}\cup \mathcal{J}^{-}_{\tau} \cup \mathcal{L}_{ \tau}.$

Given Claim~\ref{claim:exprevenue-r}, we now need to find a lower bound for $\max\big(k.\Pr[Q_\tau>k], \E[\min(P_\tau, k)] \big)$ as a function of $F(\lps, \tau)$ and $\delta_\tau$. To get this,   we start by giving lower bounds for $\E[Q_\tau]$ and $\E[P_\tau]$ in the following section, then use the expected value of these random variables to bound the value of the functions $k.\Pr[Q_\tau>k]$ and $\E[\min(P_\tau, k)]$ in Section~\ref{section:discounted-rev}. Note that all these bound will be functions of $F(\lps, \tau)$ and $\delta_\tau$.

\subsection{lower bounds for $\E[P_{\tau}]$ and $\E[Q_{\tau}]$} \label{sec:theiejdie}
In this section, we start by investigating useful facts about set $\mathcal{T}_\tau$, random variables $\b p$ and $\b q$ and the relation between them which finally leads to lower bounds for $\E[P_{\tau}]$ and $\E[Q_{\tau}]$. Let us mention that to prevent interruptions to the flow of the paper, proofs of some of the lemmas in this section are deferred to Section~\ref{app-pfs-1}.

We start by obtaining a lower bound for $\E[P_\tau]$, for which we make the three claims below.

\newcommand{\claimeqone}{The following holds:
$ \sum_{\substack{p \in \mathcal{L}_{\tau}}} \lps_{\a,p} = k\delta_{\tau}.$}
\begin{claim}\label{claim:eq1}
 \claimeqone
\end{claim}
\begin{proof}
Recalling definition \eqref{eq:defn-delta-tau},
it suffices to show 
$\bigcup_{\bu\in \B_\tau}	 \left(\mathcal{L}_\tau \cap \mathcal{Q'}_\bu\right)=  \mathcal{L}_\tau$,
as it results in  $$\delta_\tau = \sum_{\bu\in \mathcal{L}_\tau} \lps_{\a,p}/k.$$
A valid sub-profile $p=(\bu_1, \bu_2, r_1, r_2)$ is in $\bigcup_{\bu\in \B_\tau} \mathcal{Q'}_\bu$ iff $\bu_2 \in \B_\tau$ which also means $\rev{\a}{p} = \max(r_1, \bid{\a}{\bu}) \geq \tau$ and $p\in \mathcal{T}_\tau$. To complete the proof observe that this is indeed the definition of set $\mathcal{L}_\tau$ which is $\mathcal{L}_{ \tau} := \left\{(\bu', \bu, r', r)\in \mathcal{T}_{\tau} | \,\bu \in \B_\tau \right\}.$
\end{proof}

\newcommand{\claimpb}{For any buyer $\bu\in \B$ we have $$\E[p_{b, \tau}] \geq \frac{1}{\beta} ( \;\;\sum_{\mathclap{\substack{p \in \mathcal{J}^-_{\tau} \cap \mathcal{Q}_{\bu}}}} \lps_{\a,p}).$$}
\begin{claim}\label{claim:pb}
\claimpb	
\end{claim}
\begin{proof}
By construction, for vector of reserve prices $\b r$ and any buyer $\bu$ we have $$\E[p_{\bu, \tau}] = \Pr\left[r_{\bu} \in [\tau, \bid{\a}{\bu}]\right]  =  \sum_{\mathclap{r \in [\tau, \bid{\a}{\bu}]}} f_{\bu, r},$$ where $ f_{\bu, r} =x_{b,r}/\beta$ for any $r < t_{\bu}$ as defined in the algorithm. This yields that $$\E[p_{\bu, \tau}] \geq \sum_{\mathclap{\substack{r:r< t_{\bu}, \\ r \in [\tau, \bid{\a}{\bu}] }}}\, \frac{x_{b,r}}{\beta} \geq \sum_{\mathclap{\substack{r:r< t_{\bu}, \\ r \in [\tau, \bid{\a}{\bu}] }}} \;\;\;\;\;\;\, \sum_{\mathclap{p\in \mathcal{Q}_{\bu, r, \a}}} \;\, \frac{s_{\a, p}}{\beta},$$
where the second inequality is by the first and third constraints of the LP. To complete the proof it suffices to show that $$(\mathcal{J}^-_{\tau} \cap \mathcal{Q}_{\bu}) \subset \bigcup_{\mathclap{\substack{r:r< t_{\bu}, \\ r \in [\tau, \bid{\a}{\bu}] }}}\mathcal{Q}_{\bu, r, \a}.$$ Observe that we have $$(\mathcal{J}^-_{\tau} \cap \mathcal{Q}_{\bu}) = \left\{(\bu, \bu_2, r, r_2)\in \mathcal{T}_{\tau} | \, r\in \res, r_2\in \res,  \bu_2\in \B_\tau \text{ and } r < t_{\bu}\right\}.$$
Moroever, note that for any sub-profile $p=(\bu,\bu_2, r, r_2)\in\mathcal{T}_{\tau}$ we have $\rev{\a}{p} = \max(r, \bid{\a}{\bu_2})\geq\tau$, thus if $\bid{\a}{\bu_2} < \tau$ then $r \geq \tau$. As a result we have $$(\mathcal{J}^-_{\tau} \cap \mathcal{Q}_{\bu}) \subset \left\{(\bu, \bu_2, r, r_2)\in \mathcal{S}_{\a} |\, r\in \res, r_2\in \res, \bu_2\in \B_\tau \text{ and } r < t_{\bu}\right\}.$$
Further, since $\mathcal{Q}_{\bu,r,\a} := \{ (\bu, \bu_2, r, r_2) \in \mathcal{S}_{\a} | \bu_2 \in \B, r_2\in \res \},$ then
$$\bigcup_{\mathclap{\substack{r:r< t_{\bu}, \\ r \in [\tau, \bid{\a}{\bu}] }}}\mathcal{Q}_{\bu, r, \a} = \left\{(\bu, \bu_2, r, r_2) \in \mathcal{S}_\a |\, r\in \res, r_2\in \res, \bu_2\in \B,  r \geq \tau, r\geq \bid{\a}{\bu}, r<t_\bu \right\}.$$
Note that any $(\bu, \bu_2, r, r_2) \in \mathcal{S}_\a$ satisfies $r\geq \bid{\a}{\bu}$, therefore we get 
 $$(\mathcal{J}^-_{\tau} \cap \mathcal{Q}_{\bu}) \subset \bigcup_{\mathclap{\substack{r:r< t_{\bu}, \\ r \in [\tau, \bid{\a}{\bu}] }}}\mathcal{Q}_{\bu, r, \a}.\qedhere$$
 
\end{proof}

The following lemma is the last piece that we need to get the desired lower bound for $\E[P_\tau]$ in Lemma~\ref{lem:expP}. The proof of this lemma due to being lengthy is deferred to Section~\ref{app-pfs-1}.

\newcommand{\claimFS}{The following inequality holds.
$$F(\b \lps, \tau) \leq \sum_{ \mathclap{\substack{p\in \mathcal{J}^-_\tau }}} \lps_{\a, p} + \sum_{ \mathclap{\substack{p\in \mathcal{L}_\tau }}} \lps_{\a, p}.$$}
\begin{lemma}\label{lemma:FS} 
\claimFS
\end{lemma}

\begin{lemma}\label{lem:expP}
	We have the following lower bound for $\E[P_\tau]:$ $$\E[P_\tau] \geq \frac{F(\b \lps, \tau)- k\delta_\tau}{\beta}.$$ 
\end{lemma}

\begin{proof}
Based on Lemma~\ref{lemma:FS} we have $$F(\b \lps, \tau) - \sum_{ \mathclap{\substack{p\in \mathcal{L}_\tau }}} \lps_{\a, p}\leq \sum_{ \mathclap{\substack{p\in \mathcal{J}^-_\tau }}} \lps_{\a, p}.$$ Combining this by 
$ \sum_{\substack{p \in \mathcal{L}_{\tau}}} \lps_{\a,p} = k\delta_\tau$ from Claim~\ref{claim:eq1} and diving both sides by $\boost$ gives us:
$$ \frac{F(\b \lps, \tau)- k\delta_\tau}{\boost} \leq \frac{1}{\boost}\sum_{ \mathclap{\substack{p\in \mathcal{J}^-_\tau }}} \lps_{\a, p}  $$ 
We conclude the proof by noting that as a result of Claim~\ref{claim:pb}, we have $\E[P_\tau]\geq \frac{1}{\beta}\sum_{\substack{p \in \mathcal{J}^-_{\tau}}} \lps_{\a,p}.$ 
\end{proof}
Getting the desired lower bound for $\E[Q_\tau]$  is however more complicated than that of $\E[P_\tau]$. 
In Lemma~\ref{lem:expQ1} and Lemma~\ref{lem:expQ2} we give two different lower bounds for $\E[Q_\tau]$ which we then merge in Lemma~\ref{lemma:iuh4iu3} to obtain an stronger one. The proof of both these lemmas  are based on careful analysis of the relations between $\b q$ and subsets of $\mathcal{T}_\tau$, and are deferred to Section~\ref{app-pfs-1} due to being very complicated.

\newcommand{\lemexpQQQ}{For  
$\B_1 = \{\bu\in \B: \E[q_{\bu}] = 1\}$, we have $$\E[Q_\tau-|\B_1|] \geq  (k-|\B_1|+1) \delta/\beta.$$
}
\begin{lemma}\label{lem:expQ1} 
\lemexpQQQ
\end{lemma}

\newcommand{\lemexpQ}{For $\B_1 = \{\bu\in \B: \E[q_{\bu}] = 1\}$ and $m = |\B_1|$ we have  $$\E[Q_\tau-m] \geq \frac{F(\lps, \tau) - m\delta}{\beta} + \delta/\beta.$$}

\begin{lemma}\label{lem:expQ2} \lemexpQ{}
\end{lemma}

\begin{lemma}\label{lemma:iuh4iu3} For $\B_1 = \{\bu\in \B: q_{\bu, \tau} = 1\}|$ and $m = |\B_1|$, we have $$\E[Q_\tau-m] \geq \frac{ k.\max(F(\lps, \tau)/k , \delta_\tau) - (m-1) \delta_\tau}{\boost}.$$\end{lemma}
\begin{proof}
	This is a direct result of the lower bounds given in Lemma~\ref{lem:expQ1} and Lemma~\ref{lem:expQ2}, which are respectively as follows.
	$$\E[Q_\tau-m] \geq  (k-m+1) \delta_\tau/\beta.$$
		$$\E[Q_\tau-m] \geq \frac{F(\lps, \tau) - m\delta_\tau}{\beta} + \delta_\tau/\boost.$$
	By combining these lower bounds we get 
	\ascomment{moved the $\boost$ to the RHS. Much simpler!}
$$\boost\cdot\E[Q_\tau-m] 
\geq \max(F(\lps, \tau)  ,   k \delta_\tau) - (m-1)\delta_\tau  \geq k\cdot\max(F(\lps, \tau)/k,\delta_\tau)-(m-1)\delta_\tau.
	\qedhere$$
\end{proof}

\subsection{Revenue of the discounted vector}
\label{section:discounted-rev}
	In this section, we continue our effort to give a lower bound for $\E[\winners{\a}{\b r}{\tau}]$ as a function of $\delta_\tau$ and $F(\lps, \tau)$. Recall that by Claim~\ref{claim:exprevenue-r} we have $$ \E[\winners{\a}{\b r}{\tau}] \geq \max(k.\Pr[Q_\tau>k], \E[\min(P_\tau, k)]).$$ 
In Lemma~\ref{lemma:iuh4iu3} and Lemma~\ref{lem:expP} in the previous section, we have obtained lower bounds for both $\E[Q_{\tau}]$ and $\E[P_\tau]$ as functions of $\delta_\tau$ and $F(\b \lps, \tau)$. Thus, we proceed to find numeric lower bounds for $k \cdot \Pr[Q_\tau>k]$ and $\E[\min(P_\tau, k)]$ by all possible values of these parameters. To be able to do so, we use the fact that both $Q_\tau$ and $P_\tau$ are sums of Bernoulli random variables. Based on a sequence of observations about Bernoulli random variables that are mostly presented in Section~\ref{section:usefulstuff} we approximate $\Pr[Q_\tau>k]$ by a function on a set of  Bernoulli random variables whose expectation is related to  $\E[Q_\tau]$. Later, we use the relation between Binomial and Poisson distributions to get a lower bound that can be computed numerically given fixed values of $\delta_\tau$ and $F(\lps, \tau)$. We take a similar but simpler approach to find a lower bound for $\E[\min(P_\tau, k)]$.

Let us define function $G(x, \lambda)$ for a real number $\tau>0$ and any integer $x\geq0$,  as follows: \begin{equation} \label{eq:closedformq} G(x, \lambda)=1-\sum_{i=0}^{x} \frac{\lambda^i e^{-\lambda}}{i!}.\end{equation}
Note that $G(x, \lambda)$ is the probability with which a random variable drawn from $\text{Pois}(\lambda)$ is greater than  $x$. This function later arises in the lower bound for $\Pr[Q_\tau > k]$  due to the special relation between Poisson and Binomial distribution when the number of trials goes to infinity. 

We start by the following lemma about Bernoulli random variables (proved in Section~\ref{section:usefulstuff}).
\newcommand{\lemmaoiu}{
Given $m\in\mathbb{N}$ and a random variable $X$ that is sum of a set of independent Bernoulli random variables with $\E[X] = \mu$, if $m+1<\mu$, then we have $$\Pr\left[X > m\right] \geq \min_{0\leq i \leq m} G(m-i, \mu-i).$$
}
\begin{lemma}\label{lemma:oiu3obfj}
\lemmaoiu
\end{lemma}

For any $ 0 \leq i \leq k$, let us define $\lambda_i$ as follows. We have $$\lambda_0 = \min\left(\frac{2F(\lps, \tau)}{k\beta}+ \delta_\tau/\boost-2, \frac{F(\lps, \tau)}{k\beta}\right),$$  $$\lambda_1 = \min\left(\frac{2F(\lps, \tau)}{k\beta}+ \delta_\tau/\boost-1,  \frac{2F(\lps, \tau)}{k\beta}\right),$$ and for any any $i\geq 2$, $$\lambda_i = \frac{iF(\lps, \tau)}{k\beta} + \delta_\tau/\boost.$$	
We are not ready to state our lemma about a lower bound for $\frac{F(\lps, \tau)}{k\beta}.$
\begin{lemma}\label{lem:lambdaq}
If  we have $\frac{F(\lps, \tau)}{k\beta} > 1 $ and $\lambda_i\geq i+1$ for any $i\geq 0$, then 
$$\Pr[Q_\tau > k] \geq \min_{0\leq m \leq k} G(m, \lambda_m).$$
\end{lemma}

\begin{proof}
Let us define $\B_1 = \{\bu\in \B: q_\bu =  1\}$, $m=|\B_1|$, and $Q_{\tau,2} = \sum_{\bu\in \B\backslash \B_1} q_{\bu}$. We have 
$ \E[Q_\tau] =\E[Q_{\tau, 2}] + m$ which implies $\Pr[Q_\tau > k] = \Pr[Q_{\tau,2} > k - m].$
Using Lemma~\ref{lemma:iuh4iu3}, we have the following lower bound for the expected value of random variable $Q_{\tau, 2}$: \begin{equation}\label{eq:87ygvcb} \E[Q_{\tau,2}]  \geq \frac{ k.\max(F(\lps, \tau)/k , \delta_\tau) - (m-1) \delta_\tau}{\boost}.\end{equation}
Further, since $Q_{\tau,2}$ is sum of a set of independent Bernoulli random variables, if $\E[Q_{\tau,2}]\geq k-m+1$ holds, as an application of Lemma~\ref{lemma:oiu3obfj}, we get
$$\Pr[Q_{\tau, 2}> (k-m)] \geq \min_{0\leq i \leq k_1} G(k_1-i, \E[Q_{\tau,2}]-i), $$
 where $k_1= k-m$. We will later prove that $\E[Q_{\tau,2}]\geq k-m+1$ holds. Given this equation, to complete the proof it suffices to show  that for any $0\leq i\leq k_1$, we have $G(k_1-i, E[Q_{\tau,2}]-i) \geq G(k_1-i, \lambda_{k_1-i})$ since $0\leq k_1-i \leq k$. We do this by proving that the following equation holds for any $0\leq i\leq k_1$: $$\E[Q_{\tau,2}]-i \geq \lambda_{k_1-i}.$$ Note that proving this also gives us $\E[Q_{\tau,2}]\geq k-m+1$ since we get $\E[Q_{\tau,2}] \geq \lambda_{k-m}$ and by the statement of the lemma, we have  $\lambda_{k-m} \geq k-m+1$. Recall that by definition, for any $i$ with $k_1-i\geq 1$ we have $$\lambda_{k_1-i} = \frac{(k_1-i)F(\lps, \tau)}{k\boost} + \delta_\tau/\boost. $$
Moreover, using the lower bound provided for $\E[Q_{\tau,2}]$ in Equation~\ref{eq:87ygvcb}, we get $$\E[Q_{\tau,2}]\geq \frac{ k.\max(F(\lps, \tau)/k , \delta_\tau) - (m-1) \delta_\tau}{\boost} \geq \frac{ (k-m).F(\lps, \tau)/k +  \delta_\tau}{\boost} \geq  \frac{k_1F(\lps, \tau)}{k\boost} + \delta_\tau/\boost.$$
We complete the proof for the case of $k_1-i > 1$ by invoking  $\frac{F(\lps, \tau)}{k\boost} > 1$ from the statement of lemma. 
 Therefore, to complete the proof it suffices to show that $\E[Q_{\tau,2}] -i \geq \lambda_{k_1-i}$ holds for $k_1-i \leq 1$. Using Equation~\ref{eq:87ygvcb} and by the fact that $k\geq 2$, we have
 $$\E[Q_{\tau,2}]\geq \frac{ k.\max(F(\lps, \tau)/k , \delta_\tau) - m\delta_\tau + \delta_\tau}{\boost} \geq \frac{2F(\lps, \tau)/k + (k-2-m)\max(F(\lps, \tau)/k,  \delta_\tau) + \delta_\tau}{\boost}.$$ 
If $m=0$, then we have $k_1-2=k-2-m\geq 0$. Moreover, since we have $F(\lps, \tau)/(k\boost)>1$, in the case of $m=0$, we get 
 $$\E[Q_{\tau,2}]-i \geq \frac{ 2F(\lps, \tau)/k + \delta_\tau +(k_1-2) \max(F(\lps, \tau)/k,  \delta_\tau)}{\boost}-i \geq \frac{ 2F(\lps, \tau)/k+\delta_\tau}{\boost} + k_1 -2 -i. $$ This implies $\E[Q_{\tau,2}]-i \geq \lambda_{k_1-i}$ for $m=0$ and $k_1-i \leq 1$. Now, it remains to show this for $m>0$ and $k_1-i \leq 1$  as well. If $m>0$ we can write the followings:
  $$\E[Q_{\tau,2}]\geq \frac{ k.\max(F(\lps, \tau)/k , \delta_\tau) - (m-1) \delta_\tau}{\boost} \geq \frac{ (k-m+1)F(\lps, \tau)/k }{\boost} = \frac{ (k_1+1)F(\lps, \tau)/k }{\boost},$$
  $$\E[Q_{\tau,2}] - i \geq  \frac{ (k_1 - i +1 )F(\lps, \tau)/k + i\cdot F(\lps, \tau)/k}{\boost} -i \geq \frac{ (k_1 - i +1 )F(\lps, \tau)/k }{\boost} .$$
As a result of this for the cases of $(k_1-i)=0$ and  $(k_1-i)=1$  we respectively get $\E[Q_{\tau,2}] - i \geq F(\lps, \tau)/(k\boost)$ and $\E[Q_{\tau,2}] - i \geq 2F(\lps, \tau)/(k \boost)$. Knowing that by definition, we have $\lambda_0 \leq F(\lps, \tau)/(k\boost)$ and $\lambda_1 \leq 2F(\lps, \tau)/(k \boost)$ hold completes the proof.
\end{proof}
Based on a simple application of Chernoff bound, we show that for any $m\geq 2000$, we have $G(m , 1.05m) \geq 0.9$. (We will prove this as Lemma~\ref{lem:trick} in Section~\ref{section:usefulstuff}.)  Since $\lambda_m$ is an increasing function of $F(\lps, \tau)/(k\beta)$, and that $\lambda_m \geq mF(\lps, \tau)$ holds for any $m\geq 2$, this implies that for $F(\lps, \tau)/(k\beta) \geq 1.05$, and $m\geq2000$ we have $G(m , \lambda_m)\geq 0.9$. This gives us
\begin{equation}\label{eq:theop4koit4}\Pr[Q_\tau > k] \geq \min(0.9, \min_{0\leq m < 2000} G(m, \lambda_m)).\end{equation}
For smaller values of $m$; however, giving a desired lower bound for $G(m , \lambda_m)$  is unnecessarily complicated. To avoid the complication of that proof, we instead numerically compute $G(m , 
\lambda_\tau)$ for different values of $F(\lps, \tau)/(k\boost)$ and $\delta_\tau/\boost$ in Table~\ref{tab:my-table}. Then, using Lemma~\ref{lem:lambdaq} find a lower bound for $\Pr[Q_\tau > k]$ given fixed values for these variables. Each element of Table~\ref{tab:my-table}, contains value of $\min_{0\leq m < 2000} G(m, \lambda_m)$ for fixed values of $F 
(\lps, \tau)/(k\boost)$ and $\delta_\tau/\boost$. Note that, we ignore an entry of the table by inserting an $-$, if the values associated to $F(\lps, \tau)$ and $\delta_\tau/\boost$ in that entry do not satisfy the necessary conditions of Lemma~\ref{lem:lambdaq}. 
We later use this table to complete the proof of Lemma~\ref{lem:second-type}.

\renewcommand{\arraystretch}{1.5}
\begin{table}[H]
\centering
\begin{tabular}{|l||*{5}{c|}}\hline
\backslashbox{$\frac{F(\lps, \tau)}{k\beta}$ }{$\delta_\tau/\beta$}
&\makebox[4em]{0.6} &\makebox[4em]{0.8} &\makebox[4em]{0.9}
&\makebox[4em]{1} \\\hline\hline
$1.05 $  & - & - &0.57&0.59\\\hline
$1.1$   & - &0.57&0.59&0.62\\\hline
$1.2$  &0.57&0.62&0.64&0.66\\\hline

$1.5$ &0.697&0.73&0.746&0.76\\\hline
$1.7 $ &0.76&0.789&0.8&0.814\\\hline
$1.8 $ &0.789&0.8&0.826&0.834\\\hline
\end{tabular}
\caption{lower bounds for $\Pr[Q_\tau>k]$ given fixed values of $F(\lps, \tau)$ and $\delta/\beta$ based on Equation~\ref{eq:theop4koit4}.}
\label{tab:my-table}
\end{table}

To be able to use the information in this table towards giving a numeric lower bound for $\max(\Pr[Q_\tau>k], \E[\min(Q_\tau, k)]/k)$, we also need a similar table for $\E[\min(P_\tau, k)]/k$.
To provide the desired lower bound for $\E[\min(P_\tau, k)]/k$ in Lemma~\ref{lem:Pvalue}, we first need some facts about Bernoulli random variables which are stated in  Lemma~\ref{lemma:07432} and  Claim~\ref{lem:Pt} below. To prevent interruptions to the flow of this section, both proofs are deferred to Section~\ref{section:usefulstuff}.

\newcommand{\lemmazerosever}{For any integer number $m>2$ and any real number $\theta\in[0, 2]$, we have  $$\min_{\b \mu \in M_{2,\theta}} H(2, \b \mu) \leq  \min_{\b \mu \in M_{m , \theta}} H(m, \b \mu),$$ where
$M_{m , \theta} = \{ \b \mu = (\mu_1, \dots,\mu_n) \in [0, 1]^n|  \,  \sum_{i=1}^{n} \mu_i = m\theta \}, $ and 
$$H\left(k , (\mu_1, \dots, \mu_n)\right) = \frac{\E[\min (\sum_{i\in [n]} x_i, m)]}{m},$$ with $x_i, \dots, x_n$ being independent Bernoulli random variables with means $\mu_i, \dots, \mu_n$.}

\begin{lemma}\label{lemma:07432}
\lemmazerosever{}
\end{lemma}

\newcommand{\lemPt}{Given a fixed real number $\theta \in (0, 2)$, and a set of independent Bernoulli random variables $x_1, \dots, x_n$ with $\E[\sum_{i\in [n]} x_i] = 2 \theta$ we have 
$$\tfrac12\;\E[\min(\sum_{i\in [n]} x_i , 2)] \geq 1- (1+\theta)e^{-2\theta} .$$}

\begin{claim}\label{lem:Pt}\lemPt{}
\end{claim}

\begin{lemma}\label{lem:Pvalue}
For any $k>1$, we have 
\begin{equation} \label{eq:Pvalue} 
\E[\min(P_\tau, k)]/k>1- (1+\alpha)e^{-2\alpha},
\qquad\text{where}\qquad
\alpha = F(\lps, \tau)/(k\beta) - \delta/\beta.
\end{equation} 
\end{lemma}
\begin{proof}
For any real number $\theta\in[0, 2]$  we define set $$M_{k , \theta} = \{ \b \mu = (\mu_1, \dots,\mu_n) \in [0, 1]^n|  \,  \sum_{i=1}^{n} \mu_i \geq k\theta \}, $$
 and function $$H(k , (\mu_1, \dots, \mu_n)) = \frac{\E[\min (\sum_{i\in [n]} x_i, k)]}{k},$$ where $x_i, \dots, x_n$ are independent Bernoulli random variables with means $\mu_i, \dots, \mu_n$.
By Lemma~\ref{lemma:07432}, we know that  $$\min_{\b \mu \in M_{2,\theta}} H(2, \b \mu) \leq  \min_{\b \mu \in M_{k , \theta}} H(k, \b \mu).$$ 
Note that we have $(\E[p_1], \dots \E[p_n])\in M_{k, \alpha}$ since based on Lemma~\ref{lem:expP} $$\E[P_\tau] \geq \frac{F(\b \lps, \tau)- k\delta}{\beta} = \alpha k.$$ 
Moreover, observe that $H(k, (\E[p_1], \dots \E[p_n])) = \E[\min(P_\tau, k)]/k,$ which implies 
$$\E[\min(P_\tau, k)]/k> \min_{\b \mu \in M_{2 , \alpha}} G(2, \b \mu).$$
We complete the proof using Lemma~\ref{lem:Pt} that states $\min_{\b \mu \in M_{2 , \alpha}} H(2, \b \mu) \geq 1- (1+\alpha)e^{-2\alpha} .$
\end{proof}

We now proceed to construct a similar table for $\E[\min(P_\tau, k)]/k$ based on the lower bound provided in Lemma~\ref{lem:Pvalue}. Each element of Table~\ref{tab:my-table2} contains a lower bound for  $\E[\min(P_\tau, k)]/k$, given fixed values of $m 
(\lps, \tau)/(k\boost)$ and $\delta_\tau/\boost$. 
\begin{table}[H]
\centering
\begin{tabular}{|l||*{11}{c|}}\hline
$F(\lps, \tau)/(k\beta) - \delta_\tau/\beta$ &0.15&0.2&0.6&0.7&0.8&0.9&1&1.1&1.2\\\hline
$\E[\min(P_\tau, k)]/k$ &0.14&0.19&0.51&0.58&0.63&0.68&0.72&0.76&0.8\\\hline

\end{tabular}
\caption{lower bounds for $\E[\min(P_\tau, k)]/k$ based on different values of $\alpha:=F(\lps, \tau)/(\beta k) - \delta/\beta$. We use the lower bound $\E[min(P_\tau, k)]/k> 1-(1+\alpha)e^{-2\alpha}$ obtained in Lemma~\ref{lem:Pvalue}.}
\label{tab:my-table2}
\end{table}

The following lemma is the final piece that we need to complete the proof of  Lemma~\ref{lem:second-type}. In this lemma, we use the constructed tables  to show that by setting $\beta = 0.55$ in the rounding algorithm we get  $$ F(\lps, \tau)  - \beta \E[\winners{\a}{\b r}{\tau}]\leq 1.05k\beta \leq 0.58.$$

\begin{lemma}\label{lem:bestbeta}
For $\beta = 0.55$ we have 
 $  F(\lps, \tau)  - \beta \E[\winners{\a}{\b r}{\tau}]\leq 1.05k\beta.$
 \end{lemma}
 
 \begin{proof}
We start by considering different values of  $F(\lps, \tau)$ and finding a lower bound for $\E[\winners{\a}{\b r}{\tau}]$ based on that. Recall that by Claim~\ref{claim:exprevenue-r}, we have $$ \E[\winners{\a}{\b r}{\tau}]/k \geq \max\big(\Pr[Q_\tau>k], \E[\min(P_\tau, k)]/k \big)$$ 
Further, based on Lemma~\ref{lem:lambdaq} and Lemma~\ref{lem:Pvalue}, we have
$$\max\big(\Pr[Q_\tau>k], \E[\min(P_\tau, k)]/k \big) \geq \max\left(\min_{0\leq m \leq k} G(m, \lambda_m), 1-(1+\alpha)^{-2\alpha}\right),$$where $\alpha:=F(\lps, \tau)/(\beta k) - \delta_\tau/\beta$.
It is easy to see that $\min_{0\leq m \leq k} G(m, \lambda_m)$ is an increasing function of $\delta_{\tau}/\beta$, while $1-(1+\alpha)^{-2\alpha}$ is a decreasing function of $\delta_{\tau}/\beta$; therefore, for any $x\in (0, 1)$ we have\footnote{To clarify the notation, when we use $G(a, b)_{|b=x}$, for a function $G$, we refer to the value of $G$ given that $b=y$.} 
 \begin{equation}\label{eq:iuyrhf}\max(\Pr[Q_\tau>k], \E[\min(Q_\tau, k)]/k) \geq \min\left(\min_{0\leq m \leq k} G(m, \lambda_m)_{|\frac{\delta_\tau}{\beta}=x} , (1-(1+\alpha)^{-2\alpha})_{|\frac{\delta_\tau}{\beta}=x} \right).\end{equation}
 
 We use this fact to construct Table~\ref{tab:my-table3} based on Table~\ref{tab:my-table} and Table~\ref{tab:my-table2}. To do so, for any fixed value of $F(\lps, \tau)/(k\beta) = y$, we consider four possible values $x\in \{0.6, 0.8, 0.9, 1\}$ for $\delta_{\tau}/\beta$ and rewrite Equation~\ref{eq:iuyrhf} as
 \begin{align*}&\max(\Pr[Q_\tau>k], \E[\min(Q_\tau, k)]/k)_{|F(\lps, \tau)/(k\beta) = y} \\ &\geq \max_{x\in \{0.6, 0.8, 0.9, 1\}}\min\left(\min_{0\leq m \leq k} G(m, \lambda_m)_{|F(\lps, \tau)/(k\beta) = y, \frac{\delta_\tau}{\beta}=x} , (1-(1+\alpha)^{-2\alpha})_{|F(\lps, \tau)/(k\beta) = y, \frac{\delta_\tau}{\beta}=x} \right).\end{align*}
 For any $y\in\{1.05, 1.1, 1.2, 1.5, 1.7, 1.8\}$ and $x\in \{0.6, 0.8, 0.9, 1\}$, we refer to Table~\ref{tab:my-table} and Table~\ref{tab:my-table2} for values of $\min_{0\leq m \leq k} G(m, \lambda_m)_{|F(\lps, \tau)/(k\beta) = y, \delta_{\tau}/\beta=x}$  and $(1-(1+\alpha)^{-2\alpha})_{|F(\lps, \tau)/(k\beta) = y, \delta_{\tau}/\beta=x}$. (For the sake of constructing this table, if for a pair of $x$ and $y$, values of these function are not precomputed in Table~\ref{tab:my-table} and Table~\ref{tab:my-table2}, we assume that they are equal to zero.)
 \begin{table}[H]
\centering
\begin{tabular}{|l||*{8}{c|}}\hline
$F(\lps, \tau)/(k\beta)$ &1.05&1.1&1.2&1.5&1.7&1.8\\\hline
$\max(\Pr[Q_\tau>k], \E[\min(Q_\tau, k)]/k)$ &0.14&0.19&0.51&0.68&0.76&0.789\\\hline
\end{tabular}
\caption{lower bounds for $\E[\winners{\a}{\b r}{\tau}]/k$ using Equation~\ref{eq:iuyrhf}, given fixed values of $F(\lps, \tau)/(k\beta)$ achieved from combining Table~\ref{tab:my-table} and Table~\ref{tab:my-table2}.}
\label{tab:my-table3}
\end{table}

As an instance, the first column of the table is obtained as follows. For  $F(\lps, \tau)/(k\beta)= 1.05$ and $\delta_{\tau}/\beta= 0.9$, we have the followings respectively based on  Table~\ref{tab:my-table} and Table~\ref{tab:my-table2}: $$\min_{0\leq m \leq k} G(m, \lambda_m)_{|F(\lps, \tau)/(k\beta) = 1.05, \delta_{\tau}/\beta=0.9} =0.57,$$  $$\left(1-(1+\alpha)^{-2\alpha}\right)_{|F(\lps, \tau)/(k\beta) = 1.05, \delta_{\tau}/\beta=0.9}=0.14.$$ 
As a result, we have \begin{align*}&\max(\Pr[Q_\tau>k], \E[\min(Q_\tau, k)]/k)_{|F(\lps, \tau)/(k\beta) = 1.05} \\ &\geq \max_{x\in \{0.8, 0.9, 1\}}\min\left(\min_{0\leq m \leq k} G(m, \lambda_m)_{|F(\lps, \tau)/(k\beta) = 1.05, \frac{\delta_{\tau}}{\beta}=0.9} , (1-(1+\alpha)^{-2\alpha})_{|F(\lps, \tau)/(k\beta) = 1.05, \frac{\delta_{\tau}}{\beta}=0.9} \right)\\ &\geq \min(0.57, 0.14) = 0.14.\end{align*}

We now proceed to complete the proof using Table~\ref{tab:my-table3}.  
Observe that the lower bound in Equation~\ref{eq:iuyrhf} is a non-decreasing function of $F(\lps, \tau)/(k\beta)$.  This implies that for any  $0\leq x < y\leq 1$, we have 
	$ \E[\winners{\a}{\b r}{\tau} | \frac{F(\lps, \tau)}{k\beta} = x]\leq  \E[\winners{\a}{\b r}{\tau} | \frac{F(\lps, \tau)}{k\beta} \in [x, y]].$
Having this, we complete the proof using a case by case analysis and proving that $$F(\lps, \tau)  - \beta \E[\winners{\a}{\b r}{\tau}]\leq 1.05k\beta$$ holds for all possible values of $F(\lps, \tau)/(k\beta)$.
\begin{itemize}
\item 	For $F(\lps, \tau)/(k\boost) \in [0, 1.05]$, it is obvious that $F(\lps, \tau)  - \beta \E[\winners{\a}{\b r}{\tau}]\leq 1.05k\boost$ holds.
\item If $F(\lps, \tau)/(k\boost) \in [1.05, 1.1]$, then we get $F(\lps, \tau)  - \beta \E[\winners{\a}{\b r}{\tau}] \leq (1.1 - 0.14)k\boost = 0.96k\boost$ since  in this case we have $\E[\winners{\a}{\b r}{\tau}]/k \geq 0.14$ based on Table~\ref{tab:my-table3}.
\item If $F(\lps, \tau)/(k\boost) \in [1.1, 1.2]$, then we get $F(\lps, \tau)  - \beta \E[\winners{\a}{\b r}{\tau}] \leq (1.2 - 0.19)k\boost = 1.01 k\boost$ since  in this case we have $\E[\winners{\a}{\b r}{\tau}]/k \geq 0.19$.
\item If $F(\lps, \tau)/(k\boost) \in [1.2, 1.5]$, then we get $F(\lps, \tau)  - \beta \E[\winners{\a}{\b r}{\tau}] \leq (1.5 - 0.51)k\boost = 0.99 k\boost$ since  in this case we have $\E[\winners{\a}{\b r}{\tau}]/k \geq 0.51$.
\item If $F(\lps, \tau)/(k\boost) \in [1.5, 1.7]$, then we get $F(\lps, \tau)  - \beta \E[\winners{\a}{\b r}{\tau}] \leq (1.7 - 0.68)k\boost = 1.02 k\boost$ since  in this case we have $\E[\winners{\a}{\b r}{\tau}]/k \geq 0.68$.
\item If $F(\lps, \tau)/(k\boost) \in [1.7, 1.8]$, then we get $F(\lps, \tau)  - \beta \E[\winners{\a}{\b r}{\tau}] \leq (1.8 - 0.76)k\boost = 1.04 k\boost$ since  in this case we have $\E[\winners{\a}{\b r}{\tau}]/k \geq 0.76$.
\item If $F(\lps, \tau)/(k\boost) \in [1.8, 1/\boost]$, then we get $F(\lps, \tau)  - \beta \E[\winners{\a}{\b r}{\tau}] \leq (1.8 - 0.76)k\boost = 1.04 k\boost$ since  in this case we have $\E[\winners{\a}{\b r}{\tau}]/k \geq 0.76$.
\end{itemize}

The proof is concluded since we know $F(\lps, \tau)/(k\boost) \leq 1/\boost$. This is due to the definition of $F(\lps, \tau)$ which is
$$F(\b \lps, \tau) = \sum_{\mathclap{\substack{p : p\in \mathcal{S}_{\a},\\ \rev{\a}{p} \geq \tau}}} \lps_{\a, p}  - (1-\beta)\E[\winners{\a}{\b r'}{\tau}],$$
and the fact that by constraints of the LP, we have $\sum_{p : p\in \mathcal{S}_{\a}} \lps_{\a, p}  \leq k$. 
\end{proof}

\section{Upper Bound for the  ``Simple Rounding'' Approach}\label{example}
In the algorithm designed for the single unit case of the problem,   Derakhshan et al. \cite{derakhshan2019lp} provide a simple rounding technique which directly uses the fraction solution of the LP as a probability distribution over reserve prices, and for each buyer independently draws a reserve price from that. This is equivalent to setting  $\beta=0$ in our rounding algorithm. In the rest of the section, we use \emph{simple rounding} to refer to this rounding technique. We show using an example that for a large enough $k$ (number of items) the approximation factor of this algorithm is at most $0.5 + \epsilon$ for any small constant $\epsilon$. 

\begin{lemma}
Given any constant $\epsilon>0,$ there exists a dataset of bids for which the simple rounding approach achieves at most $0.5+\epsilon$ fraction of the optimal revenue.
\end{lemma}
\begin{proof}

 We use Table~\ref{tab:bad-example} to represent our dataset of auctions. Let $k$ denote the number of items. In our example, we have $2k+2$ auctions and $k+2$ buyers represented by $b_1, \dots, b_{k+2}$. Each column in the table represents an auction and its weight is the number of times that this auction is repeated in the dataset. To put it differently, we can simply assume that there are only four \emph{weighted} auctions represented by the four columns and each one has a revenue equal to the total payment of buyers in the auction multiplied by its weight. Further, note that buyers $b_3, \dots, b_{k+2}$ have similar bids in all auctions.

\renewcommand{\arraystretch}{1.5}
\begin{table}[H]
\centering
\begin{tabular}{|l||*{5}{c|}}\hline
\backslashbox{Buyer}{ Weight}
&\makebox[4em]{1} &\makebox[4em]{1} &\makebox[4em]{$k$}
&\makebox[4em]{$k$} \\\hline \hline
$\bu_1$  & $k^3$ & 0 & 0 &0\\\hline
 $\bu_2$   & 0 &0 & $k$ &1 \\\hline
 $\bu_3, \dots, \bu_{k+2}$  &0 & $k^2$ &$k$& 1 \\\hline

\end{tabular}
\caption{A Bad Example}
\label{tab:bad-example}
\end{table}

We observe that there are two different optimal vectors of reserve prices for this dataset which are $(k^3, k, k^2, \dots, k^2)$ and $(k^3, 1, 1, \dots, 1)$. By applying the first reserves, the amounts of revenue that we get from our four auctions (the four columns) are respectively   $k^3$, $k^3$, $k^2$, and 0, while the second vector gives us $k^3$, $0$, $k^3$, and $k^2$ amounts of revenue respectively. Thus, the optimal revenue adds up to $2k^3+k^2$. Further, let us mention that the all-zero vector of reserve prices simply gives us a revenue of $k^3+k^2$, and as a result the approximation factor of this solution is at most $0.5+\epsilon$ for a large enough $k$ and any constant $\epsilon>0$. Given a small constant $\epsilon>0$, we construct a fractional solution of the LP which if rounded using the simple rounding procedure, results in a less than $0.5+\epsilon$ approximate solution for large values of $k$. Let $\delta = \epsilon/4$, and consider a solution of the LP that assigns $\delta$ probability to the first optimal solution and $1-\delta$ probability to the second one. A formal representation of this solution is given below. Note that for simplicity, instead of sub-profiles, we use profiles to represent our solution. (See Definition~\ref{def:profile}.)
Also, recall that $\hat{b_1}, \dots, \hat{b_{k+1}}$ are the auxiliary buyers who bid 0 in all the auctions.
\begin{itemize}
    \item $s_{\a_1, p_1} =1 $ for $p_1 = (b_1, \hat{b_1}, \dots, \hat{b_k}, k^3, 0, \dots, 0).$
    \item $s_{\a_2, p_2} =\delta $ for $p_2 = (b_3, \dots, b_{k+2}, \hat{b_1},  , k^2, \dots, k^2, 0).$
     \item $s_{\a_2, p'_2} =1- \delta $ for $p'_2 = (b_3, \dots, b_{k+2}, \hat{b_1} , 1, \dots, 1, 0).$
      \item $s_{\a_3, p_3} =\delta $ for $p_3 = (b_2, \hat{b_1}, \dots, \hat{b_k}, k, 0, \dots, 0).$
     \item $s_{\a_3, p'_3} =1- \delta $ for $p'_3 = (b_2, b_3, \dots, b_{k+2}, 1, \dots, 1).$
     \item $s_{\a_4, p_4} =1- \delta $ for $p_4 = (b_2, b_3, \dots, b_{k+2}, 1, \dots, 1).$
     \item $s_{\a_4, p'_4} =\delta $ for $p'_4 = (\hat{b}_1, b_3, \dots, \hat{b}_{k+1}, 0, \dots, 0).$
     \item $x_{b_1, k^3} =1.$
     \item $x_{b_2, k} = \delta$, and $x_{b_2, 1} = 1- \delta$.
     \item $x_{b, k^2} = \delta$, and $x_{b, 1} = 1- \delta$ for any $b\in\{b_3, b_{k+2}\}.$
\end{itemize}

It is easy to verify that this fractional solution satisfies the constraints of the LP and is an optimal fractional solution; therefore, it suffices to show that the vector of reserve prices obtained from the simple rounding algorithm gives us at most $(0.5+\epsilon)(2k^3+k^2)$ revenue for a large enough $k$. As mentioned before, the simple rounding algorithm directly uses the fractional solution of the LP as a probability distribution to randomly pick the reserve price of any buyer. More precisely, for any buyer $b$ any reserve price $r$ is chosen with probability $x_{b,r}$ independently from other buyers. Let $\b{s}$ denote the vector of reserve prices obtained from rounding our fractional solution using this simple rounding technique. Below we investigate the revenue obtained from each auction by applying $\b{s}$.

\begin{enumerate}
    \item From the first auction we simply get revenue of $k^3$ as reserve price $k^3$ is chosen for buyer $b_1$ with probability one.
   \item In the second auction there are $k$ buyers with nonzero bids, thus all the buyers get an item and pay their reserve prices. This implies that the expected revenue of this auction is $k(\delta k^2+ (1-\delta))$ since the simple rounding algorithm chooses reserve prices of $k^2$ and $1$ for buyers $b_3, 
\dots, b_{k+2}$ with probabilities $\delta$ and $(1-\delta)$ respectively. 
\item In the third auction, however, there are $k+1$ buyers $b_2, \dots b_{k+2}$ with bid $k$; therefore, we get a revenue of $k^3$ if all the bids are cleared. Otherwise, revenue of this auction is the sum of reserve prices of the buyers whose bid is cleared multiplied by the weight of the auction (which is $k$). For any buyer $b\in \{b_3, \dots b_{k+2}\}$ the simple rounding technique chooses reserve price  $k^2$ with probability $\delta$ and reserve price 1 with probability $1-\delta$. Moreover, buyer $b_2$ gets reserve prices $k$ and $1$ respectively with probabilities $\delta$ and $1-\delta$. 
This implies that the revenue obtained from this auction is $k^3$ with probability $(1-\delta)^k$ and it is less than $2k^2$ with probability $1-(1-\delta)^k$. Note that for any constant value of $\delta$, we have $\lim_{k\rightarrow  \infty{}}{(1-\delta)^k} =0$; thus, for a large enough $k$, the expected revenue obtained from this auction is at most $2k^2+\delta$. 

\item Since the maximum bid in the fourth auction is one and its weight is $k$, the total revenue obtained from this auction is at most $k^2$.
\end{enumerate} 
By summing up the revenue obtained from the four auctions, we conclude that for a large enough  $k$, the expected total revenue is upper bounded by $(1+\delta)k^3 + 3k^2+ k$. Since $\delta = \epsilon/4$, for a large enough $k$ we have $(1+\delta)k^3 + 3k^2+ k < (1+\epsilon/2)k^3$. Recall that the optimal  vector of reserve prices in this example gives us a revenue of $2k^3+ k^2$. This implies that the approximation factor of the simple rounding algorithm is upper bounded by $0.5+\epsilon$ for a large enough $k$.
\end{proof}
\label{sec:appendix}
\section{Omitted Proofs of Section~\ref{sec:theiejdie}}
\label{app-pfs-1}

In this section, our goal is to prove Lemma~\ref{lemma:FS}, Lemma~\ref{lem:expQ1}, and Lemma~\ref{lem:expQ2} which are stated in 
Section~\ref{sec:theiejdie}. Before going into their proofs, however, we start by a series of claims that are needed for completing the proofs.

\begin{claim}\label{claim:rmorek}
If the following inequality holds, then $\E[\winners{ \a }{\b r}{\tau}] = k.$
 \begin{equation}\label{eq:rmorek}\E[\winners{ \a}{\b r'}{\tau}] < \sum_{\bu\in \B} \Pr[\bid{\a}{\bu}\geq r'_\bu \geq \tau]\end{equation}
\end{claim}

\begin{proof}
We have $$\E[\winners{ \a}{\b r'}{\tau}] = \min\left(\sum_{\bu\in \B} \Pr[\bid{\a}{\bu}\geq r'_\bu \geq \tau], k\right).$$ Define $\B' = \{ \bu\in \B:   \Pr[\bid{\a}{\bu}\geq r'_\bu \geq \tau] \neq 0\}.$ Observe that if $|\B'|\leq k$ then $$\sum_{\bu\in \B} \Pr[\bid{\a}{\bu}\geq r'_\bu \geq \tau] \leq k,$$ which implies $$\E[\winners{ \a}{\b r'}{\tau}] = \sum_{\bu\in \B} \Pr[\bid{\a}{\bu}\geq r'_\bu \geq \tau].$$
Thus, if Equation~\ref{eq:rmorek} holds then $|\B'| > k$. By Claim~\ref{lemma:hihi}, provided in the appendix, for any $\bu \in \B'$ we have $\E[q_\bu]=1$. This gives us $\Pr[\sum_{\bu\in \B'} q_\bu > k] = 1$ and  as a result $\Pr[Q_\tau > k] = 1$.  Finally, Claim~\ref{claim:exprevenue-r} gives us $\E[\winners{ \a }{\b r}{\tau}] \geq  k\Pr[Q_\tau>k],$ and concludes the proof. 
\end{proof}

\begin{assumption} \label{assumption:1} The following equation holds for any auction $\a\in \A$.
 $$\E[\winners{ \a}{\b r'}{\tau}] = \sum_{\bu\in \B} \Pr[\bid{\a}{\bu}\geq r'_\bu \geq \tau]. $$	
\end{assumption}
\begin{proof}
As a corollary of Claim~\ref{claim:rmorek}, if  $$\E[\winners{ \a}{\b r'}{\tau}] \leq \sum_{\bu\in \B} \Pr[\bid{\a}{\bu}\geq r'_\bu \geq \tau], $$ then Equation~\ref{eq:condition2} is simply satisfied for $\beta = 0.55$ and $c=0.58$ since in this case $\boost\E[\winners{ \a }{\b r}{\tau}] = \boost k$. Moreover, based on LP we have
$$\sum_{\mathclap{\substack{p : p\in \mathcal{P}_{\a}\\ \rev{\a}{p} \geq \tau}}} \lps_{\a, p} \leq k,$$ 
which results in $$\sum_{\mathclap{\substack{p : p\in \mathcal{P}_{\a}\\ \rev{\a}{p} \geq \tau}}} \lps_{\a, p}  - (1-\beta)\E[\winners{\a}{\b r'}{\tau}]   - \beta \E[\winners{\a}{\b r}{\tau}] \leq (1-\beta)k = 0.45k < 0.58k.$$ Therefore, to complete the proof of Lemma~\ref{lem:second-type}, we make the following assumption and only focus on proving  Equation~\ref{eq:condition2} for auctions that do not satisfy the mentioned condition.	
\end{proof}

\begin{claim}\label{claim:8746r8} For any buyer $\bu \in \B_\tau$ the following holds. $$\bigcup_{\bu_1 \in \B_\tau \backslash \{\bu\}} (\mathcal{Q}_{\bu, {\a}} \cap \mathcal{Q'}_{\bu_1, {\a}}) = \mathcal{Q}_{\bu, {\a}} \cap \mathcal{L}_\tau.$$
 
\end{claim}
\begin{proof}
Consider a valid profile a valid sub-profile $p=(\bu', \bu'', r', r'')\in \mathcal{S}_\a$. By definition, for any buyer $\bu_1\in \B_\tau$ we have $p
\in (\mathcal{Q}_{\bu, {\a}} \cap \mathcal{Q'}_{\bu_1, {\a}})$ iff $\bu'=\bu$ and $\bu''=\bu_1$. Moreover, due to $p$ being a valid sub-profile, it satisfies $\bu' \neq \bu''$. As a result we have $p\in \bigcup_{\bu_1 \in \B_\tau \backslash \{\bu\}} (\mathcal{Q}_{\bu, {\a}} \cap \mathcal{Q'}_{\bu_1, {\a}})$ iff $\bu' = \bu$ and $\bu''\in \B_\tau$. This completes our proof since this is equal to definition of $\mathcal{Q}_{\bu, {\a}} \cap \mathcal{L}_\tau$.
\end{proof}

\newcommand{\claimabc}{For any buyer $\bu \in \B_\tau$ we have
	$$\sum_{\mathclap{\substack{p \in \mathcal{L}_{\tau} \cap \mathcal{Q}_{\bu}}}} \lps_{\a, p} \, +\, \sum_{\mathclap{\substack{p \in \mathcal{L}_{\tau}\cap \mathcal{Q'}_{\bu}}}} \lps_{\a, p}/k \leq \delta_\tau.$$}
 
\begin{claim}\label{claim:1788}
\claimabc{} 
\end{claim}

\begin{proof}
As a constraint of the LP we have the following for  $\bu$ and any other buyer $\bu_1\in \B$. 
\begin{align} \nonumber \sum_{\mathclap{\substack{p\in \mathcal{Q}_{\bu, {\a}} \cap \mathcal{Q'}_{\bu_1, {\a}} }}} s_{\a,p} \leq \sum_{\mathclap{r\in \res}} y'_{\bu_1,r, \a}.\end{align} 
By summing both sides over all buyers in $\B_\tau \backslash \{\bu\}$ we get
\begin{align}\nonumber \sum_{\bu_1 \in \B_\tau \backslash \{\bu\}} \;\;\sum_{\mathclap{\substack{p\in \mathcal{Q}_{\bu, {\a}} \\ \cap \mathcal{Q'}_{\bu_1, {\a}} }}} s_{\a,p} \leq \sum_{\bu_1 \in \B_\tau \backslash \{\bu\}} \,\sum_{\mathclap{r\in \res}} y'_{\bu_1,r, \a}.\end{align}
Observe that by Claim~\ref{claim:8746r8} we have $$\bigcup_{\bu_1 \in \B_\tau \backslash \{\bu\}} (\mathcal{Q}_{\bu, {\a}} \cap \mathcal{Q'}_{\bu_1, {\a}}) = \mathcal{Q}_{\bu, {\a}} \cap \mathcal{L}_\tau,$$ 
which results in \begin{align}\label{eq:something71} \sum_{\mathclap{\substack{p\in  \mathcal{L}_\tau  \cap \mathcal{Q}_{\bu, {\a}}}}} s_{\a,p} \leq\, \sum_{\mathclap{\bu_1 \in \B_\tau \backslash \{\bu\}}} \;\;\;\;\sum_{\mathclap{r\in \res}} y'_{\bu_1,r, \a}.\end{align}
Moreover,  as the second constraint of the LP we have \begin{align} \nonumber y'_{\bu, r, \a} = \sum_{\mathclap{p\in \mathcal{Q'}_{\bu,r{,\a}}}} s_{\a,p}/k.\end{align} By summing up both sides of this equation over all reserve prices in $\res$, we have
\begin{align}\label{eq:somjkjrgf} \sum_{r\in \res } y'_{\bu, r, \a} = \sum_{r\in \res } 
\;\;\;\, \sum_{\mathclap{p\in \mathcal{Q'}_{b,r{,\a}}}} s_{\a,p}/k - \sum_{\mathclap{p\in \mathcal{Q'}_{\bu, {\a}}}} s_{\a,p}/k.\end{align}
Combining Equation~\ref{eq:something71} and Equation~\ref{eq:somjkjrgf} yields 
\begin{align} \sum_{\mathclap{\substack{p\in  \mathcal{L}_\tau \cap \mathcal{Q}_{\bu, {\a}}}}}  s_{\a,p} \leq \sum_{\mathclap{\bu_1 \in \B_\tau}} \;\;\;\;\sum_{\mathclap{p\in \mathcal{Q'}_{\bu_1, {\a}}}} s_{\a,p}/k - \sum_{\mathclap{p\in \mathcal{Q'}_{\bu, {\a}}}} s_{\a,p}/k .\end{align}
We further use the fact that for any buyer $\bu_2\in \B_\tau$ and any sup-profile $p= (\bu', \bu_2, r', r)\in \mathcal{Q'}_{\bu_2, \a}$ we have $p\in \mathcal{L}_\tau.$  Recall that $\mathcal{L}_{ \tau}$ is defined as $$\mathcal{L}_{ \tau} := \left\{(\bu', \bu, r', r)\in \mathcal{T}_{\tau} | \,\bu \in \B_\tau \right\}.$$ We complete the proof using the definition of $\delta_\tau$ which is $\delta_\tau := \sum_{\bu\in \B_\tau} \sum_{p \in \mathcal{L}_{\tau} \cap \mathcal{Q'}_{\bu}} \lps_{\a,p}/k.$ This gives us the followings. \begin{align} \sum_{\mathclap{\substack{p\in  \mathcal{L}_\tau  \cap \mathcal{Q}_{\bu, {\a}}}}}  s_{\a,p} \, + \; \sum_{\mathclap{\substack{p\in \mathcal{L}_{\tau} \cap \mathcal{Q'}_{\bu, {\a}} }}} s_{\a,p}/k \leq \sum_{\mathclap{\bu_1 \in \B_\tau}} \; \;\;\;\;\;\;\sum_{\mathclap{\substack{p\in \mathcal{L}_{\tau}  \cap \mathcal{Q'}_{\bu_1, {\a}} }}} s_{\a,p}/k \leq \delta_\tau.\end{align}
\end{proof}

\begin{claim}\label{lemma:hihi}
For any buyer $\bu\in \B_\tau$ with $\Pr[r'_\bu \leq \bid{\a}{\bu}]\neq 0$ we have $\E[q_\bu]=1$.
\end{claim}

\begin{proof}
By construction of $\b r'$ we have $$\Pr\left[\bid{\a}{\bu}\geq r'_\bu \right] = \sum_{\mathclap{\substack{r : r \leq \bid{\a}{\bu}, \\ r\geq t_b }}} x_{\bu, r}/(1-\boost).$$ As a result, if $\Pr[\bid{\a}{\bu}\geq r'_\bu ] \neq 0$ then $t_\bu \leq \bid{\a}{\bu}$. 
Recall that by definition of $t_\bu$, we have $$\sum_{r < t_b} x_{\bu, r} = \boost.$$ Further by construction of $\b r$ we have $$\E[q_b] = \Pr\left[r_{\bu} \leq  \bid{\a}{\bu}\right] = \sum_{\mathclap{\substack{r: r\leq \bid{\a}{\bu}, \\ r < t_b}}} x_{\bu, r}/\boost.$$
Given that $t_\bu \leq \bid{\a}{\bu}$ we get $$\E[q_b] = \Pr\left[r_{\bu} \leq  \bid{\a}{\bu}\right] = \sum_{\mathclap{r < t_b}} x_{\bu, r}/\boost=1.$$ This concludes our proof.
\end{proof}



\begin{claim}\label{claim:1935}
$$(1-\boost)\sum_{\bu\in \B} \Pr\left[\bid{\a}{\bu}\geq r'_b \geq \tau\right] \geq \sum_{ \mathclap{\substack{p\in \mathcal{J}^+_\tau }}} \lps_{\a, p}.$$ 	
\end{claim}

\begin{proof}
		By definition of $\b r'$
	\begin{align}\label{eq:i8yhruy} (1-\boost)\sum_{\bu\in \B} \Pr\left[\bid{\a}{\bu}\geq r'_b \geq \tau\right] = (1-\boost)\sum_{\bu\in \B} \;\;\;\sum_{\mathclap{\substack{r:\bid{\a}{\bu}\leq r, \\ r \leq \tau }}} f_{\bu, r} = (1-\boost)\sum_{\bu\in \B} \;\;\;\,\sum_{\mathclap{\substack{r : r\geq t_\bu, \\ r\in [\tau, \bid{\a}{\bu}]}}} \frac{x_{\bu, r}}{1-\beta} = \sum_{\bu\in \B} \;\;\;\,\sum_{\mathclap{\substack{r : r\geq t_\bu, \\ r\in [\tau, \bid{\a}{\bu}]}}} x_{\bu, r}. \end{align}
	Recall definition $ \mathcal{J}^{+}_{\tau} := \{p=(\bu, \bu', r, r')\in \mathcal{T}_{\tau} | \, \bu' \notin \B_\tau \text{, } r \geq t_{\bu'}\},$ for which we have  $$ \mathcal{J}^{+}_{\tau} \subset \left\{p=(\bu, \bu', r, r')\in \mathcal{S}_{\a} | \, r\geq \tau \text{, } r \geq t_{\bu}\right\}.$$ 
	This is because combination of $p\in\mathcal{T}_{\tau}$ and $\bu' \notin \B_\tau$ implies that revenue of the sub-profile is greater than or equal to $\tau$ while bid of the supporting buyer is smaller than $\tau$ which yields $r\geq \tau$. Moreover, due to validity of any sub-profile in $\mathcal{J}^{+}_{\tau}$ we have $\bid{\a}{\bu} \geq r$. Using the first constraint of the LP, we get $$ \sum_{ \mathclap{\substack{p\in \mathcal{J}^+_\tau }}} \lps_{\a, p} = \sum_{\bu\in \B} \;\,\sum_{ \mathclap{\substack{p\in (\mathcal{J}^+_\tau \\ \cap \mathcal{Q}_{\bu ,r})}}} \lps_{\a, p} \leq \sum_{\bu\in \B} \;\;\sum_{\mathclap{\substack{r: r \geq t_{\bu}, \\ r \in  [\tau, \bid{\a}{\bu}]}}} y_{\bu, r, \a}.$$
	Moreover, by the third constraint, we have $$ \sum_{\mathclap{\substack{r: r \geq t_{\bu},\\  r \in  [\tau, \bid{\a}{\bu}]}}} y_{\bu, r, \a} \leq \sum_{\mathclap{\substack{ r: r \geq t_{\bu},\\ r \in  [\tau, \bid{\a}{\bu}]}}} x_{\bu, r}.$$
	Evoking Equation~\ref{eq:i8yhruy}, we obtain  $$ \sum_{ \mathclap{\substack{p\in \mathcal{J}^+_\tau }}} \lps_{\a, p} \leq  \sum_{\bu\in \B} \;\;\sum_{\mathclap{\substack{ r: r \geq t_{\bu},\\ r \in  [\tau, \bid{\a}{\bu}]}}} x_{\bu, r} = (1-\boost)\sum_{\bu\in \B} \Pr\left[\bid{\a}{\bu}\geq r'_b \geq \tau\right];$$ thus, the proof is completed.
\end{proof}

\begin{claim}\label{claim:beta-delta}
The following equation holds for any buyer $\bu$.
\begin{equation}\label{eq:final} \sum_{\mathclap{\substack{p \in \mathcal{T}_{\tau} \cap \mathcal{Q}_{\bu, \a}}}} \lps_{\a, p}  \; + \;\sum_{\mathclap{\substack{p \in \mathcal{L}_{\tau}  \cap \mathcal{Q'}_{\bu, \a}}}} \lps_{\a,p}/k - (1-\boost)\Pr\left[\bid{\a}{\bu} \geq r'_\bu \geq \tau\right] 
 \leq \max(\beta, \delta_\tau).\end{equation}
\end{claim}

\begin{proof}
We consider two cases of $t_b\geq \tau$ and $t_b< \tau$ and prove the lemma for them seperately.
We show that if $t_b\geq \tau$ then the left hand side of Equation~\ref{eq:final} is uppe bounded by $\beta$, and for the second case we show that it is upper bounded by $\delta_\tau$. 
We have 
 $$\sum_{\mathclap{\substack{p \in \mathcal{T}_{\tau}\cap \mathcal{Q}_{\bu, \a}}}} \lps_{\a, p}  \; +\; \sum_{\mathclap{\substack{p \in \mathcal{L}_{\tau} \cap \mathcal{Q'}_{\bu, \a}}}} \lps_{\a,p}/k \leq \sum_{\mathclap{\substack{p\in \mathcal{Q}_{\bu, \a}}}} \lps_{\a, p}   + \sum_{\mathclap{\substack{p\in \mathcal{Q'}_{\bu, \a}}}} \lps_{\a,p}/k = \sum_{\mathclap{r\leq \bid{\a}{\bu}}}\;\;\;\;\, \sum_{\mathclap{p \in \mathcal{Q}_{\bu, r, \a}}} \lps_{\a, p} +  \sum_{\mathclap{r\leq \bid{\a}{\bu}}} \;\;\;\;\;\sum_{\mathclap{p \in \mathcal{Q'}_{\bu, r, \a}}} \lps_{\a, p}/k.$$
 The right hand side is due to the fact that any sub-profile in  $\mathcal{Q}_{\bu, \a}$ or $\mathcal{Q'}_{\bu, \a}$ is a valid profile of auction $\a$ which implies $\mathcal{Q}_{\bu, \a} = \bigcup_{r \leq \bid{\a}{\bu}}  \mathcal{Q'}_{\bu, r, \a}$ and $\mathcal{Q}_{\bu, \a} = \bigcup_{r \leq \bid{\a}{\bu}}  \mathcal{Q'}_{\bu, r, \a}$. 
  Moreover, based on the first three constraints of the LP for any $r\leq \bid{\a}{\bu}$ we have $$y_{\bu, r, \a} = \sum_{\mathclap{p \in \mathcal{Q}_{\bu, r, \a}}} \lps_{\a, p}, \;\;\;\; y'_{\bu, r, \a} = \sum_{\mathclap{p \in \mathcal{Q'}_{\bu, r, \a}}} \lps_{\a, p, \a}/k,\text{and} \;\;\;\; y'_{\bu, r, \a} + y_{\bu, r, \a} \leq x_{\bu, r}.$$ 
 which implies 
\begin{equation} \label{16576}\sum_{\mathclap{r\leq \bid{\a}{\bu}}}\;\;\;\;\; \sum_{\mathclap{p \in \mathcal{Q}_{\bu, r, \a}}} \lps_{\a, p} +  \sum_{\mathclap{r\leq \bid{\a}{\bu}}} \;\;\;\;\;\sum_{\mathclap{p \in \mathcal{Q'}_{\bu, r, \a}}} \lps_{\a, p}/k \leq \sum_{\mathclap{r\leq \bid{\a}{\bu}}} y_{\bu, r, \a} +  \sum_{\mathclap{r\leq \bid{\a}{\bu}}} y'_{\bu, r, \a}\leq \sum_{\mathclap{r\leq \bid{\a}{\bu}}} x_{\bu, r}.\end{equation}
Moreover, based on the construction of $\b r'$ we have \begin{equation} \label{eq:1567}(1-\boost)\Pr\left[\bid{\a}{\bu} \geq r'_\bu \geq \tau\right] = (1-\boost) \sum_{\mathclap{r\in [\tau, \bid{\a}{\bu}]}} f'_{\bu, r} =  (1-\beta) \sum_{\mathclap{\substack{r : r\in [\tau, \bid{\a}{\bu}],\\ r\geq t_\bu}}} x_{\bu, r}/(1-\boost).\end{equation}
If $t_\bu \geq \tau$, combining Equation~\ref{16576} and Equation~\ref{eq:1567}, gives us  
\begin{equation}\nonumber \sum_{\mathclap{\substack{p \in \mathcal{T}_{\tau} \cap \mathcal{Q}_{\bu, \a}}}} \lps_{\a, p}  \, + \, \sum_{\mathclap{\substack{p \in \mathcal{L}_{\tau} \cap \mathcal{Q'}_{\bu, \a}}}} \lps_{\a,p}/k - (1-\boost)\Pr\left[\bid{\a}{\bu} \geq r'_\bu \geq \tau\right] 
 \leq\sum_{\mathclap{r\leq \bid{\a}{\bu}}} x_{\bu, r} - \sum_{\mathclap{\substack{r\in \left[t_\bu, \bid{\a}{\bu}\right]}}} x_{\bu, r} =  \sum_{\mathclap{\substack{r < t_\bu}}} x_{\bu, r}.\end{equation}
 By definition of $t_{\bu}$, we have $\sum_{r < t_\bu} x_{\bu, r} = \beta$, therefore our proof for the case of $t_\bu \geq \tau$ is completed and in the rest of the proof we assume $t_\bu < \tau$.
 Recall definition $$\mathcal{J}^{-}_{\tau} := \{p=(\bu, \bu', r, r')\in \mathcal{T}_{\tau} | \,\bid{\a}{\bu'} < \tau \text{ and } r < t_{\bu}\}. $$ It is easy to verify that if $t_\bu < \tau$ then $\mathcal{J}^-_{\tau} \cap \mathcal{Q}_{\bu, \a}= \emptyset$. Since $\mathcal{T}_{\tau}$ is partitioned to disjoint sets of $\mathcal{J}^+_{\tau}$, $\mathcal{J}^-_{\tau}$, and $\mathcal{L}_{\tau}$ we obtain.
 $$\sum_{\mathclap{\substack{p \in \mathcal{T}_{\tau}\cap \mathcal{Q}_{\bu}}}} \lps_{\a, p} = \sum_{\mathclap{\substack{p \in \mathcal{J}^+_{\tau}\cap \mathcal{Q}_{\bu}}}} \lps_{\a, p} + \sum_{\mathclap{\substack{p \in \mathcal{L}_{\tau}\cap \mathcal{Q}_{\bu}}}} \lps_{\a, p}.$$ 
In addition, by Claim~\ref{claim:1935}, we have  $$\left(1-\boost\right)\Pr\left[\bid{\a}{\bu} \geq r'_\bu \geq \tau\right]\geq \sum_{\mathclap{\substack{p \in \mathcal{J}^+_{\tau}\cap \mathcal{Q}_{\bu}}}} \lps_{\a, p},$$ 
which gives us $$\sum_{\mathclap{\substack{p \in \mathcal{T}_{\tau}\cap \mathcal{Q}_{\bu}}}} \lps_{\a, p}  - (1-\boost) \Pr\left[\bid{\a}{\bu} \geq r'_\bu \geq \tau\right]
 \leq  \sum_{\mathclap{\substack{p \in \mathcal{L}_{\tau}\cap \mathcal{Q}_{\bu}}}} \lps_{\a, p}.$$
The proof is then completed using Claim~\ref{claim:1788} that shows $$\sum_{\mathclap{\substack{p \in \mathcal{L}_{\tau}\cap \mathcal{Q}_{\bu}}}} \lps_{\a, p} +\, \sum_{\mathclap{\substack{p \in \mathcal{L}_{\tau}\cap \mathcal{Q'}_{\bu}}}} \lps_{\a, p}/k \leq \delta_\tau$$ for any buyer $\bu$.
\end{proof}

\begin{claim}\label{claim:qb-lower} For any buyer $\bu$ with $\bid{\a}{\bu} \geq \tau$ we have 
	$$ \E[q_{\bu}] \geq \min\left(\frac{1}{\beta} \left(\;\;\;\, \sum_{\mathclap{\substack{p \in \mathcal{T}_{\tau}\cap \mathcal{Q}_{\bu}}}} \lps_{\a, p}   + \sum_{\substack{p \in \mathcal{L}_{\tau} \cap \mathcal{Q'}_{\bu}}} \lps_{\a,p}/k - (1-\boost)\Pr\left[\bid{\a}{\bu} \geq r'_\bu \geq \tau\right]\right)  , 1\right)$$
\end{claim}
\begin{proof}
We provide two different proofs for cases of $t_b \leq  \bid{\a}{\bu} $ and $t_b > \bid{\a}{\bu}$. We claim that in the first case, we have $\E[q_b] =1$ and in the second case, 	$$ \E[q_{\bu}] \geq \frac{1}{\beta} \left( \;\;\;\, \sum_{\mathclap{\substack{p \in \mathcal{T}_{\tau}\cap \mathcal{Q}_{\bu}}}} \lps_{\a, p}  + \sum_{\substack{p \in \mathcal{L}_{\tau} \cap \mathcal{Q'}_{\bu}}} \lps_{\a,p}/k - \Pr\left[\bid{\a}{\bu} \geq r_b'\geq \tau\right]\right).$$ 
By construction, for vector of reserve prices $\b r$ and any buyer $\bu$ we have $$\E[q_{\bu}] = \Pr\big[r_{\bu} \leq \bid{\a}{\bu}\big]  =  \sum_{\mathclap{r \leq \bid{\a}{\bu}}} f_{\bu, r},$$ 
where $f_{\bu, r} =x_{b,r}/\beta$ for any  $r < t_{\bu}$  and $f_{\bu, r} = 0$ for any $r\geq t_b$ as defined in the algorithm. Note that $t_b$ is chosen in a way that $\sum_{r < t_{\bu}} x_{b, r} = \boost$. Thus, if $t_b \leq \bid{\a}{\bu}$ we get
$$\E[q_{\bu}] =  \sum_{\mathclap{r \leq \bid{\a}{\bu}}} f_{\bu, r} = \sum_{\mathclap{ \substack{r: r \leq \bid{\a}{\bu}, \\ r < t_\bu}}} x_{\bu, r}/\beta = \sum_{\mathclap{ \substack{ r < t_\bu}}} x_{\bu, r}/\beta = 1.$$ This completes the proof for the first case. Therefore, in the rest of the proof we focus on the case of $t_b > \bid{\a}{\bu}$. This gives us
$$\E[q_{\bu}] =  \sum_{\mathclap{r \leq \bid{\a}{\bu}}} f_{\bu, r} = \sum_{\mathclap{ \substack{r:r \leq \bid{\a}{\bu}, \\ r < t_\bu}}} x_{\bu, r}/\beta = \sum_{\mathclap{ \substack{ r < \bid{\a}{\bu}}}} x_{\bu, r}/\beta \geq  \sum_{\mathclap{ \substack{ r < \bid{\a}{\bu}}}} \left(y_{\bu, r}+ y'_{\bu, r}\right)/\beta,$$
where the right hand side is by constraint 
$y_{\bu, r}+ y'_{\bu, r} \leq x_{\bu, r}$ in the LP. Further by the first two constraints of the LP we obtain $$\sum_{\mathclap{ \substack{ r < \bid{\a}{\bu}}}} \left(y_{\bu, r}+ y'_{\bu, r}\right) = \sum_{\mathclap{r < \bid{\a}{\bu}}}\; \left(\;\;\;\, \sum_{\mathclap{\substack{p\in \mathcal{Q}_{\bu, r, \a}}}} s_{\a, p}+ \sum_{\mathclap{\substack{p\in \mathcal{Q'}_{\bu, r, \a}}}} s_{\a, p}/k\right).$$
Note that we can drop the constraint $r < \bid{\a}{\bu}$ from the right hand side of the equation since by definition of valid sub-profiles it holds for any $p$ in $\mathcal{Q}_{\bu, \a}$ or $\mathcal{Q'}_{\bu, \a}$. This gives us $$\sum_{r < \bid{\a}{\bu}}  \left(\;\;\;\, \sum_{\mathclap{\substack{p\in \mathcal{Q}_{\bu, r, \a}}}} s_{\a, p}+ \sum_{\mathclap{\substack{p\in \mathcal{Q'}_{\bu, r, \a}}}} s_{\a, p}/k\right) = \sum_{\mathclap{\substack{p\in \mathcal{Q}_{\bu, \a}}}} s_{\a, p}+ \sum_{\mathclap{\substack{p\in \mathcal{Q'}_{\bu, \a}}}} s_{\a, p}/k \geq 
\sum_{\mathclap{\substack{p \in \mathcal{T}_{\tau}\cap \mathcal{Q}_{\bu, \a}}}} \lps_{\a, p} \, + \; \sum_{\mathclap{\substack{p \in \mathcal{L}_{\tau} \cap \mathcal{Q'}_{\bu, \a}}}} \lps_{\a,p}/k,$$ which completes our proof.
\end{proof}

\noindent \textbf{Lemma~\ref{lemma:FS}.} (restated) \claimFS

\begin{proof}
Recall that by definition $$F(\b \lps, \tau) = \sum_{\mathclap{\substack{p\in \mathcal{T}_{\a, \tau}}}} \lps_{\a, p}  - (1-\beta)\E[\winners{\a}{\b r'}{\tau}].$$ Moreover, by Assumption~\ref{assumption:1},$$\E\left[\winners{ \a}{\b r'}{\tau}\right] = \sum_{\bu\in \B} \Pr\left[\bid{\a}{\bu}\geq r'_b \geq \tau\right].$$
In addition based on Claim~\ref{claim:1935}, we know 
$$(1-\boost)\sum_{\bu\in \B} \Pr\left[\bid{\a}{\bu}\geq r'_b \geq \tau\right] \geq \sum_{ \mathclap{\substack{p\in \mathcal{J}^+_\tau }}} \lps_{\a, p}.$$ 
Considering that $\mathcal{T}_{\tau}$ is partitioned to three disjoint sets of $\mathcal{J}^+_{\tau},  \mathcal{J}^-_{\tau},$ and $\mathcal{L}_{\tau}$, by putting the mentioned equations together, we get
$$F(\b \lps, \tau)  = \sum_{\mathclap{\substack{p\in \mathcal{T}_{\a, \tau}}}} \lps_{\a, p}  - (1-\beta)\E[\winners{\a}{\b r'}{\tau}]\leq \sum_{\mathclap{\substack{p\in \mathcal{T}_{\a, \tau}}}} \lps_{\a, p}  - \sum_{ \mathclap{\substack{p\in \mathcal{J}^+_\tau }}} \lps_{\a, p}  \leq \sum_{ \mathclap{\substack{p\in \mathcal{J}^-_\tau }}} \lps_{\a, p} + \sum_{ \mathclap{\substack{p\in \mathcal{L}_\tau }}} \lps_{\a, p}.
$$\end{proof}

\noindent \textbf{Lemma~\ref{lem:expQ1}}. (restated) \lemexpQQQ

\begin{proof}
By construction of vector of reserve prices $\b r$ and $\b r'$, for any buyer $\bu$ we have $$\sum_{\mathclap{r\leq \bid{\a}{\bu}}} x_{r, \bu} = \boost \Pr[r\leq \bid{\a}{\bu}] + (1-\boost)\Pr[r\leq \bid{\a}{\bu}].$$ Moreover, by the third constraint of the LP for any  buyer $\bu$ and reserve price $r\in \res$, we have 
$y_{\bu, r, \a} + y'_{\bu, r, \a} \leq x_{\bu, r}$, which implies 
$$\sum_{\mathclap{ r\leq \bid{\a}{\bu}}} \left(y_{\bu, r, \a} + y'_{\bu, r, \a}\right) = \boost \Pr[r\leq \bid{\a}{\bu}] + (1-\boost)\Pr[r\leq \bid{\a}{\bu}].$$ 
Combining this with the first two constraints of the LP, $y_{\bu, r, \a}= \sum_{p\in \mathcal{Q}_{\bu ,r{,\a}}} s_{\a,p} $ and $ y'_{\bu, r, \a}=\sum_{p\in \mathcal{Q'}_{\bu ,r{,\a}}} s_{\a,p}$ we get
$$\sum_{\mathclap{\substack{r\leq \bid{\a}{\bu}}}}\; \left(\;\;\;\sum_{\mathclap{p\in \mathcal{Q}_{\bu ,r{,\a}}}} \lps_{\a,p} + \sum_{\mathclap{p\in \mathcal{Q'}_{\bu ,r{,\a}}}} \lps_{\a,p}/k\right) = 
\boost \Pr[r\leq \bid{\a}{\bu}] + (1-\boost)\Pr[r\leq \bid{\a}{\bu}].$$ By definition of $\mathcal{Q}_{\bu, r}$ and $\mathcal{Q'}_{\bu, \a}$ we can write  
$$\sum_{\mathclap{p\in \mathcal{Q}_{\bu {,\a}}}} \lps_{\a,p} + \sum_{\mathclap{p\in \mathcal{Q'}_{\bu {,\a}}}} \lps_{\a,p}/k = 
\boost \Pr[r\leq \bid{\a}{\bu}] + (1-\boost)\Pr[r\leq \bid{\a}{\bu}]$$ and
$$\sum_{\mathclap{p\in \mathcal{Q}_{\bu {,\a}}\cap \mathcal{L}_\tau}} \lps_{\a,p} \;+\; \sum_{\mathclap{p\in \mathcal{Q'}_{\bu {,\a}}\cap \mathcal{L}_\tau}} \lps_{\a,p}/k \leq
\boost \Pr[r\leq \bid{\a}{\bu}] \,+\, (1-\boost)\Pr[r\leq \bid{\a}{\bu}].$$
Further, by Claim~\ref{claim:1788}, for any buyer $\bu$ we have $$\sum_{\mathclap{\substack{p \in \mathcal{L}_{\tau}\cap \mathcal{Q}_{\bu}}}} \lps_{\a, p}+ \sum_{\mathclap{\substack{p \in \mathcal{L}_{\tau}\cap \mathcal{Q'}_{\bu}}}} \lps_{\a, p}/k \leq \delta_\tau.$$ 
 which means 
 $$ \sum_{\bu\in \B} \left(\;\;\;\;\; \sum_{\mathclap{p\in \mathcal{Q}_{\bu {,\a}}\cap \mathcal{L}_\tau}} s_{\a,p} \;+\; \sum_{\mathclap{p\in \mathcal{Q'}_{\bu {,\a}}\cap \mathcal{L}_\tau}} s'_{\a,p}/k\right) - \delta_{\tau}|\B_1| \leq \sum_{\bu\in \B \backslash \B_1}
\boost \Pr[r\leq \bid{\a}{\bu}] + (1-\boost)\Pr[r\leq \bid{\a}{\bu}].$$
Further, by \ref{lemma:hihi}, we know $\Pr[r\leq \bid{\a}{\bu}]=0$ holds for any $\bu \notin \B_1$. This implies
 $$ \sum_{\bu\in \B} \left( \;\;\;\;\; \sum_{\mathclap{p\in \mathcal{Q}_{\bu {,\a}}\cap \mathcal{L}_\tau}} s_{\a,p} \;+\; \sum_{\mathclap{p\in \mathcal{Q'}_{\bu {,\a}}\cap \mathcal{L}_\tau}} s'_{\a,p}/k\right) - \delta_{\tau}|\B_1| \leq \sum_{\bu\in \B\backslash \B_1}
\boost \Pr[r\leq \bid{\a}{\bu}] =  \sum_{\bu\in \B\backslash \B_1}
\boost \E[q_\bu].$$
To complete the proof, recall that we have defined $\delta_\tau = \sum_{b\in \B_\tau}\sum_{p\in \mathcal{Q'}_{\bu {,\a}}\cap \mathcal{L}_\tau} \lps_{\a,p}/k$, and by Claim~\ref{claim:eq1} we have 
$$  \sum_{\bu \in \B_\tau} \;\;\; \;\,\sum_{\mathclap{p\in \mathcal{Q}_{\bu {,\a}}\cap \mathcal{L}_\tau}} \lps_{\a,p} = \sum_{\substack{p \in \mathcal{L}_{\tau}}} \lps_{\a,p} = k\delta_{\tau}.$$ This implies 
$$\sum_{\mathclap{p\in \mathcal{Q}_{\bu {,\a}}\cap \mathcal{L}_\tau}} \lps_{\a,p} \;+\; \sum_{\mathclap{p\in \mathcal{Q'}_{\bu {,\a}}\cap \mathcal{L}_\tau}} \lps_{\a,p}/k = (k+1)\delta_\tau,$$
 $$ (k+1 - |\B_1|)\delta_\tau \leq  \sum_{\bu\in \B\backslash \B_1}
\boost \E[q_\bu] = \boost (\E[Q_\tau] - |\B_1|),$$  
 $$ (k+1 - |\B_1|)\delta_\tau/\boost \leq \left(\E[Q_\tau] - |\B_1|\right),$$
and concludes the proof.
\end{proof}

\noindent \textbf{Lemma~\ref{lem:expQ2}.} (restated) \lemexpQ{}

\begin{proof}	
Recall that by definition $$F(\b \lps, \tau) = \sum_{\mathclap{\substack{p\in \mathcal{T}_{\a, \tau}}}} \lps_{\a, p}  - (1-\beta)\E[\winners{\a}{\b r'}{\tau}].$$ 
Combining this with Asspmtion~\ref{assumption:1} gives us
  $$F(\lps, \tau) =  \sum_{\bu \in \B} \left(\;\;\;\, \sum_{\mathclap{\substack{p \in \mathcal{T}_{\tau}\cap \mathcal{Q}_{\bu}}}} \lps_{\a, p}   - (1-\boost)\Pr[\bid{\a}{\bu} \geq r'_\bu \geq \tau]\right),$$ which results in 
\begin{equation}\label{eq:something2}\frac{F(\lps, \tau) + \delta}{\beta} = \frac{1}{\beta}  \sum_{\bu \in \B} \left(\;\;\;\, \sum_{\mathclap{\substack{p \in \mathcal{T}_{\tau}\cap \mathcal{Q}_{\bu}}}} \lps_{\a, p}   + \sum_{\mathclap{\substack{p \in \mathcal{L}_{\tau} \cap \mathcal{Q'}_{\bu}}}} \lps_{\a,p}/k - (1-\boost)\Pr[\bid{\a}{\bu} \geq r'_\bu \geq \tau]\right).\end{equation}
Moreover, by Claim~\ref{claim:qb-lower}, for any buyer $\bu$, we have 
	$$ \E[q_{\bu}] \geq \min\left(\frac{1}{\beta} \;\left( \;\;\;\, \sum_{\mathclap{\substack{p \in \mathcal{T}_{\tau}\cap \mathcal{Q}_{\bu}}}} \lps_{\a, p}   + \sum_{\mathclap{\substack{p \in \mathcal{L}_{\tau} \cap \mathcal{Q'}_{\bu}}}} \lps_{\a,p}/k - (1-\boost)\Pr[\bid{\a}{\bu} \geq r'_\bu \geq \tau]\right)  , 1\right),$$ and by Claim~\ref{claim:beta-delta} for any buyer $\bu$, we have
$$\sum_{\mathclap{\substack{p \in \mathcal{T}_{\tau} \cap \mathcal{Q}_{\bu}}}} \lps_{\a, p}   + \sum_{\mathclap{\substack{p \in \mathcal{L}_{\tau} \cap \mathcal{Q'}_{\bu}}}} \lps_{\a,p}/k - (1-\boost)\Pr[\bid{\a}{\bu} \geq r'_\bu \geq \tau] 
 \leq \max(\beta, \delta_\tau).$$
 This implies that if $\delta_\tau\leq \beta$, then for any buyer $\bu$ we have $$ \E[q_{\bu}] \geq \frac{1}{\beta} \left(\;\;\;\, \sum_{\mathclap{\substack{p \in \mathcal{T}_{\tau}\cap \mathcal{Q}_{\bu}}}} \lps_{\a, p}   + \sum_{\mathclap{\substack{p \in \mathcal{L}_{\tau} \cap \mathcal{Q'}_{\bu}}}} \lps_{\a,p}/k - (1-\boost)\Pr[\bid{\a}{\bu} \geq r'_\bu \geq \tau]\right),$$ and $$\E[Q_\tau] \geq \frac{F(\lps, \tau) + \delta}{\beta}. $$
 This completes the proof for the case of $\delta_\tau\leq \beta$; therefore, in the rest of the proof we assume $\delta_\tau > \beta$ which means for any buyer $\bu$, \begin{equation} \label{eq:something1}\sum_{\mathclap{\substack{p \in \mathcal{T}_{\tau}\cap \mathcal{Q}_{\bu}}}} \lps_{\a, p}   + \sum_{\mathclap{\substack{p \in \mathcal{L}_{\tau} \cap \mathcal{Q'}_{\bu}}}} \lps_{\a,p}/k - (1-\boost)\Pr[\bid{\a}{\bu} \geq r'_\bu \geq \tau] 
 \leq \delta_\tau.\end{equation}
 Note that for any $\bu\in \B/\B_1$ we have $q_\bu < 1$ which means $$\E[q_\bu] = \frac{1}{\beta} \left(\;\;\;\, \sum_{\mathclap{\substack{p \in \mathcal{T}_{\tau}\cap \mathcal{Q}_{\bu}}}} \lps_{\a, p}   + \sum_{\mathclap{\substack{p \in \mathcal{L}_{\tau} \cap \mathcal{Q'}_{\bu}}}} \lps_{\a,p}/k - (1-\boost)\Pr[\bid{\a}{\bu} \geq r'_\bu \geq \tau]\right).$$
 Combining this with Equation~\ref{eq:something2}, we obtain $$\E[Q_\tau - m] = \sum_{\mathclap{\bu \in \B/\B_1}}\E[q_t] = \frac{F(\lps, \tau) + \delta}{\boost} - \frac{1}{\beta}\sum_{\mathclap{\bu \in \B_1}}\;  \left(\;\;\;\, \sum_{\mathclap{\substack{p \in \mathcal{T}_{\tau}\cap \mathcal{Q}_{\bu}}}} \lps_{\a, p}   + \sum_{\mathclap{\substack{p \in \mathcal{L}_{\tau} \cap \mathcal{Q'}_{\bu}}}}  \lps_{\a,p}/k - (1-\boost)\Pr[\bid{\a}{\bu} \geq r'_\bu \geq \tau]\right).$$ Using Equation~\ref{eq:something1}, we get the following which completes the proof
 $$\E[Q_\tau - m] = \sum_{\bu \in \B/\B_1}\E[q_t] \geq \frac{F(\lps, \tau) + \delta_\tau}{\boost} - \frac{m \delta_\tau}{\boost} = \frac{F(\lps, \tau) - m\delta_\tau}{\beta} + \delta_\tau/\beta.$$
\end{proof}

\section{Proof of Lemma~\ref{lemma:LP-upper}}\label{proof:lemma-lp-upper}
\noindent \textbf{Lemma~\ref{lemma:LP-upper}}. (restated)
\lemmaLPOpt
\begin{proof}
 Consider $\opt$, an optimal solution of the problem. To prove this claim, it suffices to construct vectors $\b \opts$, $\b x^o$, and $\b y^o$ in a way that setting  $\b s = \b \opts$, $\b y= \b y^o$ and $\b x= \b x^o$ in the LP satisfies all the LP constraints and that the objective function of the LP equals to the revenue of \opt, or in the other words 
\begin{equation}\label{eq:hejrbfer} \rev{}{\opt} = \max_{\b x, \b s} \sum_{\a\in \A}\sum_{p\in \mathcal{P}_{\a}} s^o_{\a,p} \cdot \rev{\a}{p}.\end{equation}
Roughly speaking, we construct $\b \opts$ to be the representation of $\opt$ in the profile space. For any sub-profile $p = (\bu_1, \bu_2, r_1, r_2)$ we have $\opts_{a, p} = 1$ iff using $\opt$, in auction $\a$ buyer $\bu_1$ is one of the winners, buyer $\bu_2$ is the supporting buyer, and their reserve prices are respectively $r_1$ and $r_2$. 

We first show that Equation~\ref{eq:hejrbfer} holds for $\b s^o$. Let $\b r$ denote the vector of reserve prices in \opt. For any auction $\a\in \A$, let $\bu_{s, \a}$ be the supporting buyer, and let $\B_\a$ denote the set of winners in auction $\a$ using vector of reserve prices $\b r$.  Payment of any winner $\bu \in \B_\a$ in auction $\a$ is $\max(r_{\bu}, \bid{\a}{\bu_{s, \a}}),$ which means revenue obtained from auction $\a$ using the vector of reserve prices $\b r$ is $\sum_{\bu \in \B_\a} \max(r_{\bu}, \bid{\a}{\bu_{s, \a}}).$ To prove Equation~\ref{eq:hejrbfer}, it suffices to show that for any auction $\a$ we have $$\sum_{\bu \in \B_\a} \max(r_{\bu}, \bid{\a}{\bu_{s, \a}}) = \sum_{p\in \mathcal{P}_\a} \rev{\a}{p}.$$
Note that for any profile $p=(\bu, r'_1, \bu_2, r'_2)\in \mathcal{P}_\a$ we have $s_{\a, p} =1 $ iff $\bu\in \B_{\a}$, $\bu_2 = \bu_{s, \a}$, $r_{\bu} =r'_1$ and $r_{\bu_2} =r'_2$. Moreover, by defintion, we have $\rev{\a}{p}= \max(r'_1, \bid{\a}{\bu_2}) = \max(r_{\bu}, \bid{\a}{\bu_{s, \a}})$. This implies that $$\sum_{p\in \mathcal{P}_\a} \rev{\a}{p} = \max(r_{\bu}, \bid{\a}{\bu_{s, \a}}),$$ which results in Equation~\ref{eq:hejrbfer}.

To complete the proof we construct $\b x^o$, $\b y^o$, and $\b y'^o$ in a way that setting $\b x= \b x^o$,  $\b y'= \b y'^o$, $\b y= \b y'^o$ and $\b s = \b \opts$, satisfied all the constraints of the LP. \begin{itemize}
 \item 	For any buyer $\bu\in \B$ and $r\in \res$ we set $x^o_{\bu, r} = 1$ iff reserve price $r$ is assigned to buyer $\bu$ in $\opt$ and set $x^o_{\bu, r} = 0$ otherwise. 
 \item For any buyer $\bu\in \B$, auction $\a\in \A$, and reserve price $r\in \res$ we set $y^o_{\bu, r, \a}=1$ iff using solution $\opt$, buyer $\bu$ is a winner in auction $\a$ and he is assigned a reserve price $r$. Otherwise we set $y^o_{\bu, r, \a}=0$.
 \item  For any buyer $\bu\in \B$, auction $\a\in \A$, and reserve price $r\in \res$ we set $y'^o_{\bu, r, \a}=1$ if using solution $\opt$, buyer $\bu$ is the supporting buyer in auction $\a$ and he is assigned a reserve price $r$. Otherwise we set $y'^o_{\bu, r, \a}=0$.
 \end{itemize}
We now start investigating the constraints of the LP one by one and verify that all seven constraints hold for  $\b x= \b x^o$,  $\b y'= \b y'^o$, $\b y= \b y'^o$ and $\b s = \b \opts$. 
\begin{enumerate}
	\item For the first constraint we need to show $y^o_{\bu, r, \a} = \sum_{p\in \mathcal{Q}_{\bu ,r{,\a}}} s^o_{\a,p}$ for any $\bu\in \B, r\in \res$, and $\a\in \A$. It is easy to see that if $\bu$ is not a winner in auction $\a$ then both sides of this equation are equal to zero. To complete the proof we assume that $\bu$ is a winner. Let $\bu_2$ denote the supporting buyer in auction $\a$. Moreover, let $r_1$ and $r_2$ respectively denote the reserve prices assigned to buyers $\bu$ and $\bu_2$ in $\opt$. We have $y^o_{\bu, r, \a} = 1$ iff $r=r_1$ and $\bu$ is a winner in auction $\a$. We also show that $\sum_{p\in \mathcal{Q}_{\bu ,r{,\a}}} s^o_{\a,p} =1$ iff $r=r_1$. Note that $p = (\bu, r_1, \bu_2, r_2)$ is the only sub-profiles in $\mathcal{S}_\a$ with $s^o_{\a, p}=1$. Further, by definition $\mathcal{Q}_{\bu ,r{,\a}}$ is a subset of $\mathcal{S}_\a$ and contains  the sub-profiles in which buyer $\bu$ is a winner in auction $\a$, thus it contains $p$ iff $r= r_1$. This means that both sides of the equation are equal to one if $r=r_1$ and both are equal to zero otherwise.
\item For the second constraint, we need to show that $y'^o_{\bu, r, \a} = \sum_{p\in \mathcal{Q'}_{\bu ,r{,\a}}} s^o_{\a,p}/k$ holds for any $\bu\in \B, r\in \res$, and $\a\in \A$. 
Let $\bu_1$ be the supporting buyer in auction $\a$ and let $r_1$ be the reserve price assigned to this buyer.
By definition we have $y'^o_{\bu, r, \a}=1$ iff $\bu =\bu_1$ and $r=r_1$, and we have $y'^o_{\bu, r, \a}=1$ otherwise.  Therefore, for this constraint to be satisfied we need to show that $\sum_{p\in \mathcal{Q'}_{\bu ,r{,\a}}} s^o_{\a,p} = k$ holds iff $\bu =\bu_1$ and $r=r_1$, and we have $\sum_{p\in \mathcal{Q'}_{\bu ,r{,\a}}} s^o_{\a,p} =0$ otherwise.  Let $\bu_2$ denote one of the winners in auction $\a$ and let $r_2$ be the reserve price assigned to this buyer. By definition, for any profile $p=(\bu_2, r_2, \bu_1, r_1)\in \mathcal{S}_\a$ we have $s^o_{\a, p} = 1$ iff $\bu_2$ is a winner in auction $\a$ and $r_1$ is the reserve price assigned to him. Since we have assumed that each auction has exactly $k$ winners then we have $\sum_{p\in \mathcal{Q'}_{\bu ,r{,\a}}} s^o_{\a,p} = k$ iff $\bu =\bu_1$ and $r=r_2$, and we have $\sum_{p\in \mathcal{Q'}_{\bu ,r{,\a}}} s^o_{\a,p} = 0$ otherwise.
\item To prove that our constructed solution satisfies the third constraint of the LP, we need to show $y^o_{\bu, r, \a} + y'^o_{\bu, r, \a} \leq x^o_{\bu, r}$ for any $\bu\in \B$, any reserve price $r\in \res$ and any auction $\a$. Consider a buyer $\bu$ and let $r_1$ be the reserve price assigned to this buyer in solution $\opt$. For any $r\neq r_1$ both sides of the equation are obviously zero. However, for $r=r_1$ we have $x^o_{\bu, r} =1$. Observe that in any auction $\a$, we also have $y^o_{\bu, r, \a} + y'^o_{\bu, r, \a}\leq 1$ since $\bu$ cannot be both a winner and the supporting buyer in an auction.
\item For the fourth constraint, we need to show $\sum_{p\in \mathcal{Q}_{\bu_2, {\a}}  \cap \mathcal{Q'}_{\bu_1, {\a}}} s^o_{\a,p} \leq \sum_{r\in \res} y'^o_{\bu_1,r, \a}.$ For any sub-profile $p = (\bu_3, r_1, \bu_4, r_2)$ in $\mathcal{Q}_{\bu_2, {\a}}  \cap \mathcal{Q'}_{\bu_1, {\a}}$ we have $\bu_4 = \bu_2$ and $\bu_3 = \bu_1$. Moreover, we have $s_{\a, p} =1$ iff $r_1$ and $r_2$ are respectively the reserve prices assigned to buyers $\bu_1$ and $\bu_2$, buyer $\bu_1$ is a winner in auction $\a$ and buyer $\bu_2$ is the supporting buyer in this auction. This implies that the left hand side is equal to one iff $\bu_1$ and $\bu_2$ are respectively a winner and the supporting buyer in auction $\a$. Further, the right hand side is equal to one iff $\bu_1$ is the supporting buyer in auction $\a$. This concludes that the fourth constraint is satisfied for $\b s^o$ and $\b y'^o$.
\item The condition $\sum_{p \in \mathcal{P}_{\a}} s^o_{\a,p} \leq k$ is satisfied due to the assumption that we have $k$ winners in all auctions. Let $\bu'$ be the supporting buyer in auction $\a$. For a sub-profile $(\bu_1, r_1, \bu_2, r_2) \in \mathcal{P}_{\a}$ we have $s^o_{\bu, r} =1$ iff  $\bu_2=\bu'$, $\bu_1$ is a winner in auction $\a$, and $r_1$ and $r_2$ are respectively the reserve prices assigned to buyers $\bu_1$ and $\bu_2$. Given that we have exactly $k$ winner and that each buyer has a unique reserve price then, $\sum_{p \in \mathcal{P}_{\a}} s^o_{\a,p} \leq k$.
\item For this constraint, we need to show $\sum_{r\in \res} x^o_{\bu,r} = 1$. This is  true since for a reserve $r$ we have $x^o_{\bu,r} = 1$ if reserve $r$ is assigned to buyer $\bu$, and in $\opt$ there is exactly one reserve price assigned to each buyer.
\item The last constraint is simply satisfied since $s^o_{\a,p}$ is either zero or one.
\end{enumerate}
\end{proof}

\section{Useful Facts about Bernoulli Random Variables} \label{section:usefulstuff}
\begin{lemma}~\label{lemmamath1}
If $Y\sim \text{Binomial}(n, p)$, then for any $m\leq n$ we have $\Pr(Y\geq m) = G(p),$ where $$G(p) = \frac{n!}{(m-1)!(n-m)!} \int_{1-p}^1 t^{n-m} (1-t)^{m-1}dt.$$
\end{lemma}
\begin{proof}
Let us define $H(p):=\Pr(Y\geq m).$ We have $$H(p)= \sum_{j=m}^n { n \choose j} p^j (1-p)^{n-j}.$$ By taking  derivative of this function we get:
\begin{align*}
	H'(p)&=\sum_{j=m}^n \binom nj jp^{j-1}(1-p)^{n-j}
	-\sum_{j=m}^n \binom nj (n-j)p^j(1-p)^{n-j-1} \\
	&=n\sum_{j=m}^n \binom{n-1}{j-1} p^{j-1}(1-p)^{n-j}
	-n\sum_{j=m}^{n-1} \binom{n-1}j p^j(1-p)^{n-j-1} \\
	&=n\sum_{i=m-1}^{n-1} \binom{n-1}i p^i(1-p)^{n-1-i}
	-n\sum_{i=m}^{n-1} \binom{n-1}i p^i(1-q)^{n-1-i} \\
	&=n\binom{n-1}{m-1} p^{m-1}(1-p)^{n-m} = G'(p).
\end{align*}
 The proof is completed as we also have $G(0)= H(0)=0.$ 
\end{proof}

\begin{lemma} \label{lemma:mathover}
For any $\alpha >1$ and $m>0$, $\Pr[X\geq m]$ is minimized subject to $X \sim \text{Binomial}(n, m\alpha/n)$ when $n\rightarrow \infty$.
\end{lemma}
\begin{proof}
Let us define $X_n \sim \text{Binomial}(n, m\alpha/n)$ and $X_{n+1} \sim \text{Binomial}(n, m\alpha/(n+1))$
We need to show that $\Pr[X_n\geq m] \leq \Pr[X_{n+1}\geq m]$. decreasing function of $n$.
Using Lemma~\ref{lemmamath1}, we have

 $\Pr[X_n\geq m] = \frac{n!}{(m-1)!(n-m)!}G(n)$ where 
$$G(n, \alpha) = \frac{n!}{(m-1)!(n-m)!} \int_{1-m\alpha/n}^1 t^{n-m} (1-t)^{m-1}dt.$$
We have $$\frac{ \Pr[X_{n+1}\geq m]}{\Pr[X_n\geq m]} = \frac{(n-m+1)G(n+1)}{(n+1)G(n, \alpha)}.$$ Define $D_{n, \alpha} = (n-m+1)G(n+1)- (n+1)G(n).$ To complete the proof it suffices to show  $D_{n, \alpha} \leq 1.$

\newcommand{\al}{\alpha}
\newcommand{\pd}[2]{\frac{\partial #1}{\partial #2}}
\newcommand{\De}{D}

\begin{align} \nonumber
	&\pd{\De_{n,\al}}\al \, \al (m\al)^{-m} (n+1)^{m-1} \left(1-\frac{m\al}{n}\right)^{m-n}\\ \nonumber
&	= \De_{n,\al;1}:=
\left(1-\frac{m\al}{n}\right)^{m-n} \left(1-\frac{m\al}{n+1}\right)^{n-m+1}-\frac{n-m+1}{n}\,\left(\frac{n+1}{n}\right)^{m-1}.  
\end{align} 
Moreover, for any $\al\in(1,n/m)$, 
\begin{align}\nonumber
	&\pd{\De_{n,\al;1}}\al
	=\frac{(\al-1) m^2  n^{n-m} (n+1-m\al)^{n-m} }{(n+1)^{n-m+1} (n-m\al)^{n-m+1}} >0. 
\end{align}
As a result, sign of $\De_{n,\al;1}$ can only change from $-$ to $+$ as $\al$ increases from $1$ to $n/m$. So, $\pd{\De_{n,\al}}\al$ has the same sign pattern. To get $D_{n, \alpha} \leq 1$, it suffices to show that $\De_{n,0}\le0$ and $\De_{n,n/m}\le0$. 

Since for $\alpha =0$, we have $G(n,0)=0$ for all $n$, we obviously have $\De_{n,0}=0\le0$. We conclude the proof by noting the following.
\begin{align} \nonumber
	& \De_{n,n/m}=(n+1)G(n+1,n/m)-(n-m+1)G(n,n/m)\\ \nonumber
 & \le(n+1)G(n+1,(n+1)/m)-(n-m+1)G(n,n/m)\\ \nonumber
&=1\Big/\binom n{m-1}-1\Big/\binom n{m-1}=0.
\end{align}
\end{proof}

\begin{lemma}\label{lemma:extreme}

Let $\b p = (p_0,  \dots, p_n) \in [0,1]^n$ be an extreme point of function $F_m (\b p)$ defined below subject to $\sum_{i\in [n]} p_i=t$ for a fixed $t$. $$F_{m}(\b p) = \Pr\Big[\sum_{i \in [n]} x_i > m\Big],$$ where $x_1, \dots, x_n$ are independent Bernoulli random variables with $\E[x_i]=p_i$ for any $i\in [n]$. 
 The following holds for any $i, j\in [n]$. If $p_i\notin \{0, 1\}$ and $p_j \notin \{0, 1\}$, then $p_i = p_j$. 
\end{lemma}

\begin{proof}
We use proof by contradiction. If there does not exist such an extreme point, then let $\b p$ be an arbitrary extreme point with maximum $u_{\b p}$ defined as below. $$u_{\b p}=  \min_{i\in U}(p_i)$$ where $U= \{i: p_i\neq 0\}$. W.l.o.g, let $p_i = u_{\b p}$ and pick a $j\in [n]$ where $p_i < p_j$.
 To obtain a contradiction, we show that if $\b p$ is an extreme point, then it is possible to modify $p_i$ and $p_j$ without changing other elements of $\b p$ in a way that the values of $p_i+p_j$ and $F_m (\b p)$ are unchanged but $p_i$ is increased. This gives us a contradiction since by repeating this process one can increase  $u_{\b p}$.
 
  Let $X=\sum_{l\in n} x_l$ and $X' = X-x_i- x_j$. We have $$F_m (\b p) = \Pr[X'>m]+\Pr[X'=m]\Pr[x_i+x_j>0] + \Pr[X'=m-1]\Pr[x_i+x_j=2].$$
Define $d=(p_j - p_i)/2$ and $s= (p_i + p_j)/2$. We have $\Pr[x_i+x_j=2] = (s+d)(s-d) = s^2-d^2$ and $\Pr[x_i+x_j>0] = d^2+2s-s^2$. Let $$G(s, d) = \Pr[X'=m] (d^2+2s-s^2) + \Pr[X'=m-1](s^2-d^2).$$ 
Therefore, 
$$F_m (\b p) =\Pr[X'>m]+ G(s, d).$$
Note that $\frac{\partial G}{\partial d} = 2d(\Pr[X'=m] + \Pr[X'=m-1])$. By the assumption that $\b p$ is an extreme point, $p_i\notin \{0, 1\}$ and $p_j \notin \{0, 1\}$ we obtain that $\Pr[X'=m] + \Pr[X'=m-1] = 0$. This gives us the freedom to change values of $p_i$ and $p_j$ and set $p_i=p_j =s$ as it does not change the value of $F_m (\b p)$. Thus, we obtain a contradiction and the proof is completed.

\end{proof}

\begin{fact}\label{fact:thefact}
Given $n$ iid Bernoulli random variables $x_1, \dots, x_n$ with $\E[x_i] = p$, if $n\rightarrow \infty$, then for any $0\leq j\leq n$ we have $$\Pr\left[\sum_{i=1}^n x_i = j\right] = \frac{e^{-np}(np)^j}{j!}.$$
\begin{proof}
This is based on the relation between Poisson and Binomial distribution when 	$n\rightarrow \infty$.
\end{proof}

\end{fact}

\noindent \textbf{Lemma~\ref{lemma:oiu3obfj}}. (restated)
\lemmaoiu

\begin{proof}
Let $x_1 \dots x_n$ denote a set of independent Bernoulli random that minimize $\Pr[\sum_{i=1}^n x_i > m]$ subject to $\sum_{i=1}^n x_i = \E[X]$.
By Lemma~\ref{lemma:extreme}, we know that any two variables $x_i$ and $x_j$ that are not deterministically zero or one are identical. Let $I = \{i\in[n]: \E[x_i]=1\}$. Moreover, w.l.o.g., assume none of the variables are deterministically zero. We have $\Pr\left[\sum_{i=1}^n x_i > m\right]  = \Pr[\sum_{i\notin M} x_i > m-|I|].$ Further by Lemma~\ref{lemma:mathover}, $\Pr[\sum_{i=1}^n x_i > m]$ is minimized when $n \rightarrow \infty$, which using Fact~\ref{fact:thefact}, leads to $$\Pr\left[\sum_{i\notin M} x_i \leq m - |I|\right] \leq \sum_{j=0}^{m-|I|} \frac{(\mu-|I|)^j e^{-(\mu-|I|)}}{j!}.$$ Recall that we have \begin{equation} \nonumber G(x, \lambda)=1-\sum_{i=0}^{x} \frac{\lambda^i e^{-\lambda}}{i!}.\end{equation}
Note that if $|I|> m$, then $\Pr[\sum_{i=1}^n x_i > m] = 1$, thus by considering all possible values of $0\leq |I| \leq m$ we get  $$\Pr\left[\sum_{i=0}^n x_i \geq m \right] \geq \min_{0\leq i\leq m}\left(1-\sum_{i=0}^{m-i} \frac{(\mu-i)^j e^{-(\mu-i)}}{j!}\right) = \min_{0\leq i \leq m} G(m-i, \mu-i).$$
\end{proof}

\noindent \textbf{Lemma~\ref{lemma:07432}.} (restated)
\lemmazerosever

\begin{proof}
Given a $\theta\in [0, 1]$ and an arbitrary $m\geq 2$, let $\b x = (x_1, \dots x_n)$ be a vector of independent Bernoulli random variables with means $\mu_1, \dots, \mu_n$ summing to $m\theta$ and let $ \b y = (y_1, \dots y_n)$ be a vector of independent Bernoulli random variables with expectations $\frac{m-1}{m}\mu_1,  \dots, \frac{m-1}{m}\mu_n$ summing up to $(m-1)\theta$. We will prove that for any such $\b x$ and $\b y$ we have \begin{equation} \label{eq:math1}\E\left[\frac{\min(\sum_{i=1}^n x_i, m)}{m}\right] \geq \E\left[\frac{\min(\sum_{i=1}^n y_i, m-1)}{m-1}\right]. \end{equation}
This shows that for a given $m$ and $\theta$, $$\min_{\b p \in M_{m , \theta}} H(m, \b p)$$ is an increasing function of $m$ and as a result for any $m$ we have  $$\min_{\b p \in M_{2,\theta}} H(2, \b p) \leq  \min_{\b p \in M_{m , \theta}} H(m, \b p).$$
To achieve this, for any $i$, we couple $x_i$ with $y_i$ so that 
$\E[y_i| x_i=0] = 0$ and $\E[y_i| x_i=1] = \frac{m-1}{m},$ and show that \begin{equation}\label{eq:math2}\E\left[K(\frac{m-1}{m}x_i, \dots, \frac{m-1}{m}x_n)\right] \geq \E \left[K(y_1, \dots, y_n)\right],\end{equation} where $K(z_1, \dots, z_n) = \min(\sum_{i=1}^n \frac{z_i}{m-1}, 1).$ Note that the left hand side of Equation~\ref{eq:math2}, is equal to that of  Equation~\ref{eq:math1} and similarly the right hand sides are equal too, thus proving the correctness of Equation~\ref{eq:math2} would complete the proof. Let $\bar{\b{x}} = $ be an arbitrary realization of the vector of random variables $\b x$. We show that for any value of $\bar{\b{x}}$ we have 
\begin{equation} \label{eq:math3} \E\left[K(\frac{m-1}{m}x_i, \dots, \frac{m-1}{m}x_n) |\b x = \bar{\b x} \right] \geq \E\left [K(y_1, \dots, y_n)|\b x = \bar{\b x}\right].\end{equation}
We know that if at least $m$ elements of $\bar{\b x}$ are equal to one, then the left hand side is equal to one as well. Also, function $K(.)$ is upper bounded by one thus, if $\sum_{i=1}^n \bar{x}_i > m-1$, the Equation~\ref{eq:math3} holds. Otherwise, if at most $m-1$ elements of $\bar{\b x}$ are equal to one then we have 
 \begin{equation} \nonumber \E\left[K(\frac{m-1}{m}x_i, \dots, \frac{m-1}{m}x_n) |\b x = \bar{\b x},  \sum_{i=1}^n \bar{x}_i \leq m-1 \right] = \min\left(\sum_{i=1}^n \frac{\bar{x}_i}{m}, 1\right) =  \frac{\sum_{i=1}^n\bar{x}_i}{m}. \end{equation}
  Recall that for any $y_i$ we have $\Pr[y_i|x_i=0] =0$, which results in $$\Pr\left[\sum_{i=1}^n y_i \leq m-1 | \b x = \bar{\b x} ,  \sum_{i=1}^n \bar{x}_i \leq m-1\right] = 1.$$ 
 This implies that  \begin{equation} \nonumber \E\left[K(y_i, \dots, y_n) |\b x = \bar{\b x},  \sum_{i=1}^n \bar{x}_i \leq m-1 \right] = \E\left[\min\left(\sum_{i=1}^n \frac{{y}_i}{m-1}, 1\right)\right] =  \E\left[\sum_{i=1}^n \frac{{y}_i}{m-1} \right] =  \frac{\sum_{i=1}^n\bar{x}_i}{m}, \end{equation} which completes the proof as we showed that Equation~\ref{eq:math3} holds for all possible realizations of $\b{x}$.

\end{proof}

\noindent \textbf{Claim~\ref{lem:Pt}}. (restated)
 \lemPt

\begin{proof}
Let $X = \sum_{i\in [n]} x_i$. We have $$\E[\min(X , 2)] = \Pr[X=1] + 2\Pr[X\geq 2] = 2 - \Pr[X=1] - 2 \Pr[X=0] $$
It is easy to verify that this function is minimized when variables $x_1, \dots, x_n$ are iid and $n\rightarrow \infty$ in which case, $X$ is a Poisson random variable with $\lambda = 2\theta$. For $X\sim \text{Poisson}(2\theta)$ we have $$\Pr[X=1] + 2 \Pr[X=0] = 2e^{-2\theta} + 2\theta^{-2\theta} = (2+2\theta)e^{-2\theta},$$
which results in $$\frac{\E[\min(X , 2)]}{2} = 1-(1+\theta)e^{-2\theta}.$$
\end{proof}

\begin{lemma} \label{lem:trick}
Given a real number $\alpha > 1.05$, an integer $m \geq 2000$, and a set of independent Bernoulli random variables $x_i, \dots, x_n$ with $\E[X] \geq \alpha m$ we have 
$$\Pr[X \geq m+1 ] \geq 0.9,$$
where $X= \sum_{i=1}^n x_i$.
\end{lemma}

\begin{proof}
Note that for any $m \geq 2000$, 	we have $1.05m < 1.049(m+1)$; therefore, given that $\alpha \geq 1.05$ and $m\geq 2000$ we obtain $\E[X] \geq 1.049 (m+1).$ By Chernoff bound, for any $\delta>0$, we have $$\Pr\left(X < (1-\delta) \E[X]\right) < \left(\frac{e^{-\delta}}{(1-\delta)^{1-\delta}}\right)^{E[X]}.$$
In our case, we are interested in giving an upper bound for $\Pr[X < m+1]$, which is

$$\Pr\left[X < m+1\right] \leq \Pr\left[X< \frac{1}{1.049} \E[X]\right] < \left(\frac{e^{-0.047}}{(1-0.047)^{1-0.047}}\right)^{2000*1.05} < 0.1.$$  
This gives us $\Pr[X \geq m+1] > 0.9$ and completes the proof.
\end{proof}

\bibliographystyle{alpha}
\bibliography{ref}
	
\end{document}